\documentclass[aos]{imsart}

\RequirePackage{amsthm,amsmath,amsfonts,amssymb, mathrsfs, upgreek, enumerate, mathtools, comment, graphicx, mathabx, booktabs, threeparttable, tikz, relsize, exscale, epstopdf, scrextend, multirow, pdfpages, bm}
\RequirePackage[numbers]{natbib}
\RequirePackage[colorlinks,linkcolor=black,citecolor=black,urlcolor=black,bookmarksopen=true,pdfstartview=FitH]{hyperref}

\startlocaldefs

\theoremstyle{plain}
\newtheorem{theorem}{Theorem}

\newtheorem{proposition}{Proposition}

\newtheorem{sectheorem}{Theorem}[section]

\newtheorem{seclemma}{Lemma}[section]

\theoremstyle{remark}
\newtheorem{assump}{Assumption}
\newtheorem{secassump}{Assumption}[section]

\newtheorem{remark}{Remark}

\newtheorem{secdefinition}{Definition}[section]

\newtheorem{secexample}{Example}[section]

\def\Snospace~{\S{}}

\def\indep{\perp\!\!\!\perp}
\newcommand{\argmax}{\operatornamewithlimits{argmax}}

\newcommand{\cov}{\text{Cov}}
\newcommand{\var}{\text{Var}}
\newcommand{\E}{{\bf E}}
\newcommand{\R}{\mathbb{R}}
\newcommand{\F}{\mathcal{F}}
\newcommand{\N}{\mathcal{N}}
\newcommand{\C}{\mathcal{C}}
\newcommand{\prob}{{\bf P}}
\newcommand{\plimarrow}{\stackrel{p}\longrightarrow}
\newcommand{\dlimarrow}{\stackrel{d}\longrightarrow}

\providecommand{\abs}[1]{\lvert#1\rvert} 
\providecommand{\norm}[1]{\lVert#1\rVert}

\newcommand*{\medcap}{\mathbin{\scalebox{1.5}{\ensuremath{\cap}}}}

\let\emptyset\varnothing
\providecommand{\abs}[1]{\lvert#1\rvert} 
\providecommand{\norm}[1]{\lVert#1\rVert}

\usepackage{pdfpages} 
\usepackage{pgffor} 

\makeatletter
\AtBeginDocument{\let\LS@rot\@undefined}
\makeatother

\def\supplementfilename{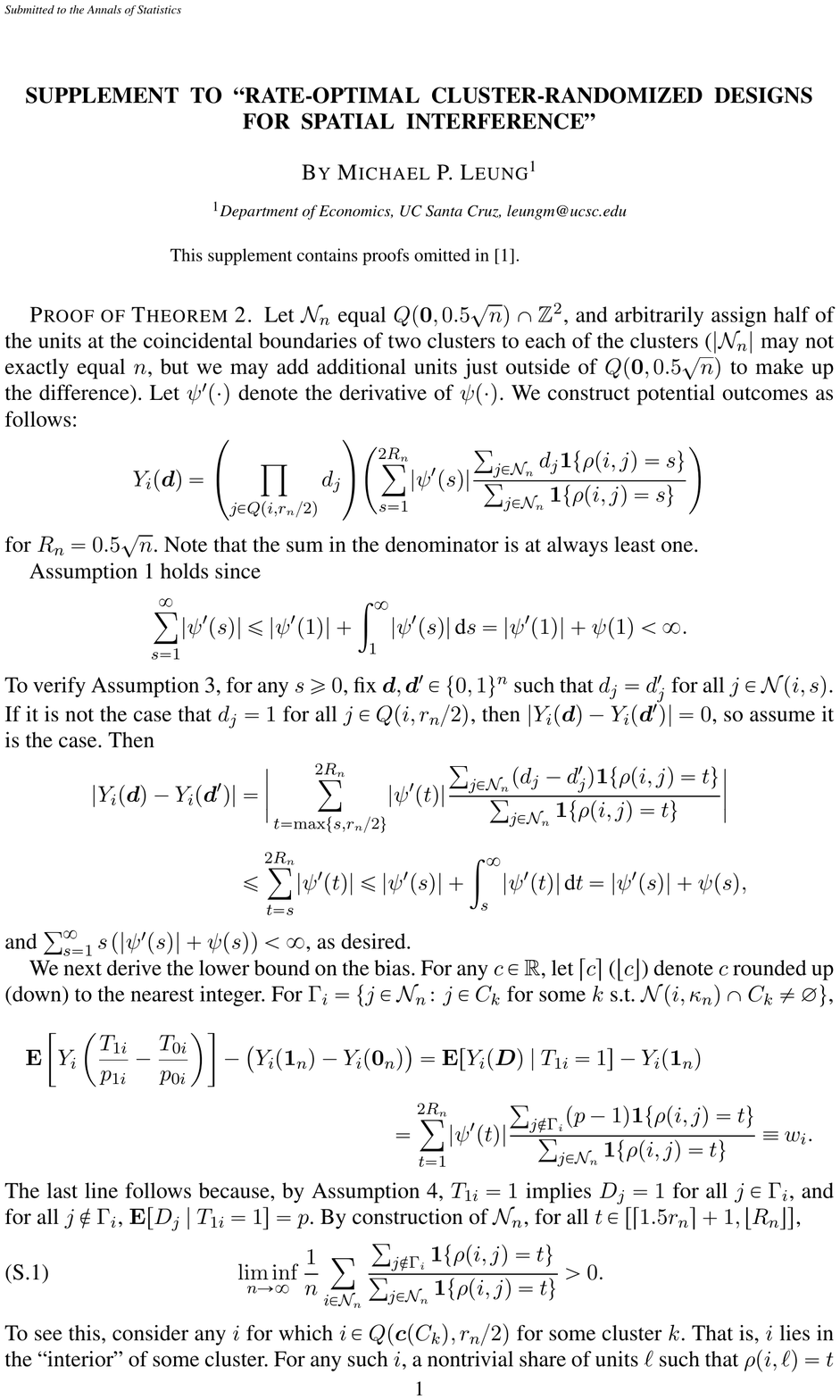}

\pdfximage{\supplementfilename}
\def\numbersupplementpages{\the\pdflastximagepages}

\newif\ifarXiv
\arXivtrue 

\endlocaldefs

\begin{document}

\begin{frontmatter}
\title{Rate-Optimal Cluster-Randomized Designs for Spatial Interference}
\runtitle{Rate-Optimal Designs for Interference}

\begin{aug}
\author[A]{\fnms{Michael P.} \snm{Leung}\ead[label=e1]{leungm@ucsc.edu}}
\address[A]{Department of Economics, UC Santa Cruz, \printead{e1}}
\end{aug}

\begin{abstract}
  We consider a potential outcomes model in which interference may be present between any two units but the extent of interference diminishes with spatial distance. The causal estimand is the global average treatment effect, which compares outcomes under the counterfactuals that all or no units are treated. We study a class of designs in which space is partitioned into clusters that are randomized into treatment and control. For each design, we estimate the treatment effect using a Horvitz-Thompson estimator that compares the average outcomes of units with all or no neighbors treated, where the neighborhood radius is of the same order as the cluster size dictated by the design. We derive the estimator's rate of convergence as a function of the design and degree of interference and use this to obtain estimator-design pairs that achieve near-optimal rates of convergence under relatively minimal assumptions on interference. We prove that the estimators are asymptotically normal and provide a variance estimator. For practical implementation of the designs, we suggest partitioning space using clustering algorithms. 
\end{abstract}



\end{frontmatter}


\section{Introduction}\label{sintro}

Consider a population of $n$ experimental units. Denote by $Y_i(\bm{d})$ the potential outcome of unit $i$ under the counterfactual that the population is assigned treatments according to the vector $\bm{d} = (d_i)_{i=1}^n \in \{0,1\}^n$, where $d_i=1$ ($d_i=0$) implies unit $i$ is assigned to treatment (control). Treatments assigned to alters can influence the ego since $Y_i(\bm{d})$ is a function of $d_j$ for $j\neq i$, what is known as {\em interference}. 

An important estimand of practical interest is the {\em global average treatment effect}
\begin{equation*}
  \theta_n = \frac{1}{n} \sum_{i\in\N_n} \big( Y_i(\bm{1}_n) - Y_i(\bm{0}_n) \big) 
\end{equation*}

\noindent where $\bm{1}_n$ ($\bm{0}_n$) is the $n$-dimensional vector of ones (zeros). This compares average outcomes under the counterfactuals that all or no units are treated. Each average can only be directly observed in the data under an extreme design that assigns all units to the same treatment arm, which would necessarily preclude observation of the other counterfactual. Common designs used in the literature, including those studied here, assign different units to different treatment arms, so neither average is directly observed in the data. Nonetheless, we show that asymptotic inference on $\theta_n$ is possible for a class of cluster-randomized designs under spatial interference where the degree of interference diminishes with distance.

Many phenomena diffuse primarily through physical interaction. The government of a large city may wish to compare the effect of two different policing strategies on crime, but more intensive policing in one neighborhood may displace crime to adjacent neighborhoods \citep{blattman2021place,verbitsky2012causal}. A rideshare company may wish to compare the performance of two different pricing algorithms, but these may induce behavior that generates spatial externalities, such as congestion. Other phenomena exhibiting spatial interference include infectious diseases in animal \citep{donnelly2003impact} and human \citep{miguel2004worms} populations, pollution \citep{giffin2020generalized}, and environmental conservation programs \citep{paler2015social}.

Much of the existing literature assumes that interference is summarized by a low-dimensional exposure mapping and that units are individually randomized into treatment or control either via Bernoulli or complete randomization \citep[e.g.][]{aronow_estimating_2017,basse2019randomization,forastiere2020identification,manski2013identification,toulis2013estimation}. Jagadeesan et al.\ \cite{jagadeesan2020designs} and Ugander et al.\ \cite{ugander2013graph} also utilize exposure mappings but depart from unit-level randomization. They propose new designs that introduce cluster dependence in unit-level assignments in order to improve estimator precision. We build on this literature by (1) studying rate-optimal choices of both cluster-randomized designs and Horvitz-Thompson estimators, (2) avoiding exposure mapping restrictions on interference, which can be quite strong \citep{eckles2017design}, and (3) developing a distributional theory for the estimator and a variance estimator. 

Regarding (2), most exposure mappings used in the literature imply that only units within a small, known distance from the ego can interfere with the ego's outcome. We instead study a weaker restriction on interference similar to \cite{leung2021causal}, which states that the degree of interference decreases with distance but does not necessarily zero out at any given distance. This is analogous to the widespread use of mixing-type conditions in the time series and spatial literature instead of $m$-dependence because the latter rules out interesting forms of autocorrelation, including models as basic as AR$(1)$. 

Regarding (1), we study cluster-randomized designs in which units are partitioned into spatial clusters, clusters are independently randomized into treatment and control, and the assignment of a unit is dictated by the assignment of its cluster. By introducing correlation in assignments, such designs can avoid overlap problems common under Bernoulli randomization, which improves the rate of convergence. For analytical tractability, we focus on designs in which clusters are equally sized squares, each design distinguished by the number of such squares. We pair each design with a Horvitz-Thompson estimator that compares the average outcomes of units with all or no treated neighbors, where the neighborhood radius is of the same order as the cluster size dictated by the design. See \autoref{nbhds} for a depiction of a hypothetical design and neighborhoods used to construct the estimator. 

Our results inform how the analyst should choose the number of clusters (and hence, the cluster size and neighborhood radius of the estimator) to minimize the rate of convergence of the estimator. Notably, existing work on cluster randomization with interference utilizes small clusters (those with asymptotically bounded size). We show that such designs are generally asymptotically biased under the weaker restriction on interference we impose, which motivates the large-cluster designs we study.

\begin{figure}
  \centering
  \includegraphics[scale=0.6]{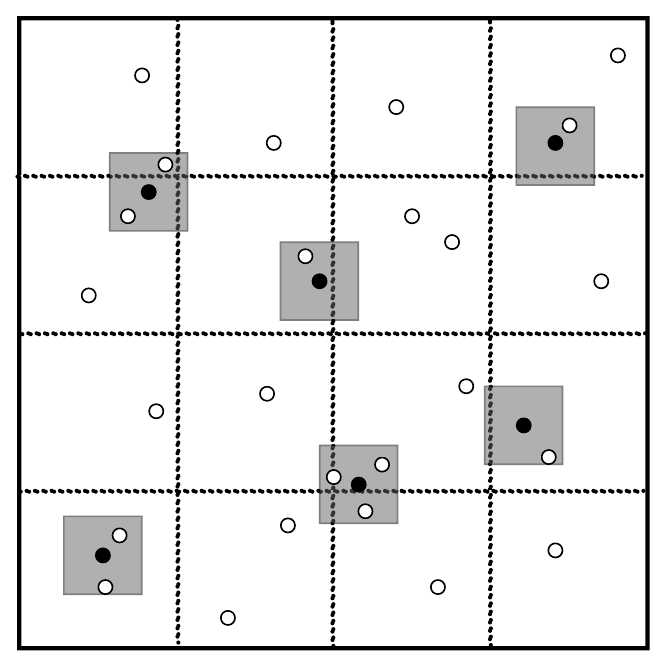}
  \caption{White and black dots depict units. White squares correspond to clusters and gray squares to neighborhoods of black units used to construct the estimator.}
  \label{nbhds}
\end{figure}

Finally, regarding (3), we show that the estimator is asymptotically normal and provide a variance estimator. These results appear to be novel, as no existing central limit theorems seem to apply to our setup in which treatments exhibit cluster dependence, clusters can be large, and units in different clusters are spatially dependent due to interference. As usual, the variance estimator is biased due to heterogeneity in unit-level treatment effects. However, we show that, in a superpopulation setting in which potential outcomes are weakly spatially dependent, the bias is asymptotically negligible.

Based on our theory, we provide practical recommendations for implementing cluster-randomized designs in \autoref{sprac}. Of course, rate-optimality results do not determine the choice of nonasymptotic constants that are often important in practice under smaller sample sizes. Still, they constitute an important first step toward designing practical procedures. Due to the generality of the setting, which imposes quite minimal assumptions on interference, it seems reasonable to first study rate-optimality, as finite-sample optimality appears to require substantial additional structure on the problem. We note that existing results on graph cluster randomization, which require stronger restrictions on interference than this paper, are nonetheless limited to rates, and how ``best'' to construct clusters in practice has been an open question.

\subsection{Related Literature}

Most of the literature supposes interference is mediated by a network. Studying optimal design in this setting is difficult because network clusters can be highly heterogeneous in topology, and their graph-theoretic properties can closely depend on the generative model of the network \citep{leung2021network}. We study spatial interference, and to make the optimal design problem analytically tractable, we focus on a class of designs that partitions space into equally sized squares while exploring in simulations the performance of more realistic designs that partition using clustering algorithms. We discuss in \autoref{sconclude} the (pessimistic) prospects of extending our approach to network interference.

There is relatively little work on optimal experimental design under interference. Viviano \cite{viviano2020experimental} proposes variance-minimizing two-wave experiments under network interference. Baird et al.\ \cite{baird2018optimal} study the power of randomized saturation designs under partial interference. 

A recent literature studies designs for interference that depart from unit-level randomization. A key paper motivating our work is \cite{ugander2013graph}, who propose graph cluster randomization designs under network interference. Ugander and Yin \cite{ugander2020randomized} study a new variant of these designs, and \cite{harshaw2021design} consider related designs for bipartite experiments. These papers assume interference is summarized by exposure mappings, which enables the construction of unbiased estimators and use of designs in which clusters are small. Under our weaker restriction on interference, we show that large clusters are required to reduce bias, which creates a bias-variance trade-off. 

Eckles et al.\ \cite{eckles2017design} show that graph cluster randomization can reduce the bias of common estimators for $\theta_n$ in the absence of correctly specified exposure mappings. Pouget-Abadie et al.\ \cite{pouget2018optimizing} propose two-stage cluster-randomized designs to minimize bias under a monotonicity restriction on interference. Several papers \citep{basse2018model,jagadeesan2020designs,sussman2017elements} study linear potential outcome models and propose designs targeting the direct average treatment effect, rather than $\theta_n$. Under a normal-sum model, \cite{basse2018model} compute the mean-squared error of the difference-in-means estimator, which they use to suggest model-assisted designs. 

The aforementioned papers on cluster randomization target global effects such as $\theta_n$ \citep[also see][]{chin2019regression,choi2017estimation}. Much of the literature on interference considers fundamentally different estimands defined by exposure mappings. When these mappings are misspecified, the estimands are functions of assignment probabilities, in which case their interpretations can be specific to the experiments run \citep{savje2021causal,savje2017average}. Hu et al.\ \cite{hu2022average} (\S 5) views this as ``largely unavoidable'' in nonparametric settings with interference. Our results show that inference on $\theta_n$, which avoids this issue, is possible under restrictions on interference weaker than those typically used in the literature. Additionally, papers in the literature impose an overlap assumption, which implicitly restricts the estimand \citep{leung2021causal}. We study cluster-randomized designs that directly satisfy overlap.

There is a large literature on cluster-randomized trials \citep[e.g.][]{hayes2017cluster,park2020assumption}. This literature predominantly studies partial interference, meaning that units are divided into clusters such that those in distinct clusters do not interfere. That is, the clusters themselves impose restrictions on interference. In our setting, clusters are determined by the design and do not restrict interference.

Finally, \cite{aronow2020design}, \cite{pollmann2020causal}, and \cite{zigler2021bipartite} study spatial interference in a different ``bipartite'' setting in which treatments are assigned to units or locations that are distinct from the units whose outcomes are of interest. This shares some similarities with spatial cluster randomization, where different spatial regions are randomized into treatment, so some of the ideas here may be applicable to optimal design there.

\subsection{Outline}

The next section defines the model of spatial interference and the class of designs and estimators studied. In \autoref{srmse}, we derive the estimator's rate of convergence, discuss rate-optimal designs, and provide practical design recommendations. In \autoref{sinfer}, we prove that the estimator is asymptotically normal, propose a variance estimator, and characterize its asymptotic properties. We report results from a simulation study in \autoref{smc}, exploring the use of spectral clustering to implement the designs. Finally, \autoref{sconclude} concludes.

\section{Setup}\label{smodel}

Let $\N_n$ be a set of $n$ units. We study experiments in which units are cluster-randomized into treatment and control, postponing to \autoref{sdesign} the specifics of the design. For each $i\in\N_n$, let $D_i$ be a binary random variable where $D_i=1$ indicates that $i$ is assigned to treatment and $D_i=0$ indicates assignment to control. Let $\bm{D} = (D_i)_{i\in\N_n}$ be the vector of realized treatments and $\bm{d} = (d_i)_{i\in\mathcal{N}_n} \in \{0,1\}^n$ denote a non-random vector of counterfactual treatments. Recall from \autoref{sintro} that $Y_i(\bm{d})$ is the {\em potential outcome} of unit $i$ under the counterfactual that units are assigned treatments according to $\bm{d}$. Formally, for each $n\in\mathbb{N}$ and $i\in\N_n$, $Y_i(\cdot)$ is a non-random function from $\{0,1\}^n$ to $\R$. We denote $i$'s factual, or observed, outcome by $Y_i = Y_i(\bm{D})$ and maintain the standard assumption that potential outcomes are uniformly bounded.

\begin{assump}[Bounded Outcomes]\label{aboundY}
  $\sup_{n\in\mathbb{N}} \max_{i\in\N_n} \max_{\bm{d}\in\{0,1\}^n} \abs{Y_i(\bm{d})} < \infty$.
\end{assump}

\subsection{Spatial Interference}\label{sspaint}

Thus far, the model allows for unrestricted interference in the sense that $Y_i(\bm{d})$ may vary essentially arbitrarily in any component of $\bm{d}$. In order to obtain a positive result on asymptotic inference, it is necessary to impose restrictions on interference to establish some form of weak dependence across unit outcomes. The existing literature primarily focuses on restrictions captured by {\em $K$-neighborhood exposure mappings}, which imply that $D_j$ can only interfere with $Y_i(\bm{D})$ if the distance between $i,j$ is at most $K$. We will discuss how this assumption is potentially restrictive and establish results under weaker conditions.

We assume each unit is located in $\mathbb{R}^2$. Label each unit by its location, so that $\N_n \subset \mathbb{R}^2$, and equip this space with the sup metric $\rho(i,j) = \max_{t=1,2} \abs{i_t-j_t}$ for $i=(i_1,i_2)$, $j=(j_1,j_2)$, and $i,j\in\mathbb{R}^2$. Let $Q(i,r) = \{j\in\mathbb{R}^2\colon \rho(i,j) \leq r\}$, the ball of radius $r$ centered at $i$. Under the sup metric, balls are squares, and the radius is half the side length of the square. Letting $\bm{0}$ denote the origin, we consider a sequence of {\em population regions} $\{Q(\bm{0},R_n)\}_{n\in\mathbb{N}}$ such that 
\begin{equation*}
  \N_n \subset Q(\bm{0},R_n), \quad\text{where}\quad R_n n^{-1/2} \rightarrow c \in (0,\infty).
\end{equation*}

\noindent That is, units are located in the square $Q(\bm{0},R_n)$ with growing radius $R_n$. Combined with the next increasing domain assumption, the number of units in the region is $O(n)$, but throughout, we will simply assume the number is exactly $n$.

\begin{assump}[Increasing Domain]\label{aid}
  There exists $\rho_0>0$ such that, for any $n\in\mathbb{N}$ and $i,j\in\N_n$, $\rho(i,j) \geq \rho_0$.
\end{assump}

\noindent This allows units to be arbitrarily irregularly spaced, subject to being minimally separated by some distance $\rho_0$, a widely used sampling framework in the spatial literature \citep[e.g.][]{jenish2009central}. In contrast, ``infill'' asymptotic approaches that do not require minimal separation and instead assume increasingly dense sampling from a fixed region can yield nonstandard limiting behavior \citep{lahiri1996inconsistency}. For some applications, the spatial distribution of units may exhibit ``hotspots'' with unusually high densities, perhaps making the infill approach more plausible. Some work adopts a hybrid of the two approaches \citep{lahiri2003central,lahiri2006resampling}, and it may be possible to extend our results to this framework.

Let $\R_+$ denote the set of non-negative reals and
\begin{equation*}
  \N(i,K) = Q(i,K) \cap \N_n
\end{equation*}

\noindent denote the {\em $K$-neighborhood} of $i$. We study the following model of interference similar to that proposed by \cite{leung2021causal}.

\begin{assump}[Interference]\label{aani}
  There exists a non-increasing function $\psi\colon \R_+ \rightarrow \R_+$ such that $\psi(0) \in (0,\infty)$, $\sum_{s=1}^\infty s\, \psi(s) < \infty$, and, for all $s\in\mathbb{R}_+$,
  \begin{equation*}
    \sup_{n\in\mathbb{N}} \max_{i\in\N_n} \max\big\{\abs{Y_i(\bm{d}) - Y_i(\bm{d}')}\colon \bm{d},\bm{d}'\in\{0,1\}^n, d_j=d_j' \,\,\forall j\in\N(i,s)\big\} \leq \psi(s). 
  \end{equation*}
\end{assump}

\noindent To interpret this, observe that $\max\big\{\abs{Y_i(\bm{d}) - Y_i(\bm{d}')}\colon \bm{d},\bm{d}'\in\{0,1\}^n, d_j=d_j' \,\,\forall j\in\N(i,s)\big\}$ maximizes over pairs of treatment assignment vectors that fix the assignments of units in $i$'s $s$-neighborhood but allow assignments to freely vary outside of this neighborhood. It therefore measures the degree of spatial interference in terms of the maximum change to $i$'s potential outcome caused by manipulating treatments assigned to units $k$ ``distant'' from $i$ in the sense that $\rho(i,k) > s$. The assumption requires the degree of interference to vanish with the neighborhood radius $s$ so that treatments assigned to more distant alters interfere less with the ego. The rate at which interference vanishes is controlled by $\psi(s)$, which is required to decay at a rate faster than $s^{-2}$. 

\begin{remark}[Necessity of rate condition]\label{rrate}
  Assumption 3(b) of \cite{jenish2009central} and Assumption 4(c) of \cite{jenish2012spatial} impose the same minimum rate of decay as \autoref{aani} on various measures of spatial dependence (mixing or near-epoch dependence coefficients) to establish central limit theorems. If the rate is slower, then the variance can be infinite. For example, consider a spatial process $\{Z_i\}_{i\in\mathcal{N}_n}$ such that units are positioned on the integer lattice $\mathbb{Z}^2$ and $\cov(Z_i,Z_j) = f(\rho(i,j))$ for some function $f(\cdot)$ for any $i,j$. Then
  \begin{equation*}
    \var\left( \frac{1}{\sqrt{n}} \sum_{i=1}^n Z_i \right) = \frac{1}{n} \sum_{i=1}^n \sum_{s=0}^\infty \sum_{j: \rho(i,j)=s} \cov(Z_i,Z_j) = \sum_{s=0}^\infty f(s) \frac{1}{n} \sum_{i=1}^n \abs{\{j\colon \rho(i,j)=s\}}.
  \end{equation*}

  \noindent Note that $\abs{\{j\colon \rho(i,j)=s\}} \leq 8s$, with equality achieved for all $i$ not near the boundary of the population region. Thus, for large $n$, a finite variance requires that $f(s)$ decay with $s$ faster than $s^{-2}$.
\end{remark}

We next discuss two models of interference satisfying \autoref{aani}. The first is the standard approach of specifying a $K$-neighborhood exposure mapping. Such a mapping is given by $T_i = T(i,\bm{D},\N_n)$ with the crucial property that its dimension does not depend on $n$, unlike that of $\bm{D}$. The approach is to assume that the low-dimensional $T_i$ summarizes interference by reparameterizing potential outcomes as
\begin{equation}
  Y_i(\bm{D}) = \tilde Y_i(T_i). \label{cspec}
\end{equation}

\noindent That is, once we fix $i$'s exposure mapping $T_i$, its potential outcome is fully determined. No less important, it is also typically assumed that exposure mappings are restricted to a unit's $K$-neighborhood, where $K$ is small, meaning fixed with respect to $n$. Formally, $T(i,\bm{d},\N_n) = T(i,\bm{d}',\N_n)$ for any $\bm{d},\bm{d}'$ such that $d_j=d_j'$ for all $j \in \N(i,K)$, which implies that the treatment assigned to a unit $j$ only interferes with $i$ if $\rho(i,j) \leq K$. In practice, choices with $K=1$ are most common, for example $T_i = (D_i,S_i)$ for $S_i = \bm{1}\{\sum_{j\in\mathcal{N}_n} G_{ij}T_j > 0\}$ or $S_i = \sum_{j\in\mathcal{N}_n} G_{ij}T_j$ where $G_{ij} = \bm{1}\{\rho(i,j)\leq 1\}$. In these examples, $D_i$ captures the direct effect of the treatment, and $S_i$ captures interference from units near $i$.

Crucially, $T_i$ and $K$ must be known to the analyst in this approach, which is often a strong requirement. In contrast, \autoref{aani} enables the analyst to impose \eqref{cspec} {\em while requiring neither to be known}. Indeed, if there exists a $K$-neighborhood exposure mapping satisfying \eqref{cspec}, then \autoref{aani} holds with $\psi(s) = c\,\bm{1}\{s \leq K\}$ for some $c$ sufficiently large. 

Furthermore, \autoref{aani} allows for more complex forms of interference ruled out by \eqref{cspec} in which interference decays more smoothly with distance, rather than being truncated at some distance $K$. The former is analogous to mixing conditions, which are widespread in the time series and spatial literature, while the latter is analogous to $m$-dependence, which rules out interesting forms of autocorrelation, including models as basic as AR$(1)$. 

In the spatial context, our assumption accommodates, for example, the Cliff-Ord autoregressive model \citep{cliff1973spatial,cliff1981spatial}, which is a workhorse model of spatial autocorrelation used in a variety of fields, including geography \citep{getis2008history}, ecology \citep{valcu2010spatial}, and economics \citep{anselin2001spatial}. A typical formulation of the model is
\begin{equation*}
  Y_i = \alpha + \lambda \sum_{j\in\mathcal{N}_n} W_{ij}Y_j + D_i\beta + \varepsilon_i, 
\end{equation*}

\noindent where we assume $\varepsilon_i$ is uniformly bounded to satisfy \autoref{aboundY}. Let $\bm{W}$ be the $n\times n$ spatial weight matrix whose $ij$th entry is $W_{ij}$. These weights typically decay with distance $\rho(i,j)$ in a sense to be made precise below. While this model is highly stylized, the important aspect it captures is autocorrelation through the spatial autoregressive parameter $\lambda$. If this is nonzero, then there is no $K$-neighborhood exposure mapping for which \eqref{cspec} holds, a point previously noted by \cite{eckles2017design} in the context of network interference.

To see this, first note that coherency of the model requires nonsingularity of $\bm{I} - \lambda \bm{W}$, where $\bm{I}$ is the $n\times n$ identity matrix. Let $\bm{V}$ be the inverse of this matrix and $V_{ij}$ its entry corresponding to units $(i,j)$. Then the reduced form of the model is
\begin{equation}
  Y_i(\bm{D}) \equiv Y_i = \sum_{j\in\mathcal{N}_n} V_{ij}(\alpha + D_j\beta + \varepsilon_j), \label{clifford}
\end{equation}

\noindent a spatial ``moving average'' model with spatial weight matrix $\bm{V}$. (See \autoref{smc} for some examples of $\bm{V}$.) Noticeably, $Y_i(\bm{D})$ can potentially depend on $D_j$ for any $j\in\mathcal{N}_n$, which is ruled out if one imposes a $K$-neighborhood exposure mapping. 

Outcomes satisfying \eqref{clifford} are near-epoch dependent, a notion of weak spatial dependence, when the weights decay with spatial distance in the following sense:
\begin{equation}
  \sup_{n\in\mathbb{N}} \max_{i\in\N_n} \sum_{j\in\N_n} \abs{V_{ij}} \rho(i,j)^\gamma < \infty \label{coani}
\end{equation}

\noindent for some $\gamma>0$ \citep[see Proposition 5 and eq.\ (13) of][]{jenish2011spatial}. The next result shows that this condition is sufficient for verifying \autoref{aani} if $\gamma>2$.

\begin{proposition}\label{pclifford}
  Suppose potential outcomes are given by \eqref{clifford} and spatial weights satisfy \eqref{coani} for some $\gamma > 2$. Then \autoref{aani} holds with $\psi(s) = c\,\min\{s^{-\gamma},1\}$ for some $c \in (0,\infty)$ that does not depend on $s$.
\end{proposition}
\begin{proof}
  Fix $s\geq 0$ and any $\bm{d},\bm{d}'\in\{0,1\}^n$ such that $d_j=d_j'$ for all $j\in\N(i,s)$. For $s<1$, $\abs{Y_i(\bm{d}) - Y_i(\bm{d}')} \leq c_1 \equiv 2\sup_n \max_i \max_{\bm{d}} \abs{Y_i(\bm{d})}$. For $s\geq 1$,
  \begin{equation*}
    \abs{Y_i(\bm{d}) - Y_i(\bm{d}')} \leq \abs{\beta} \sum_{j\in\N_n} \abs{V_{ij}} \abs{d_j - d_j'} \bm{1}\{\rho(i,j) > s\} \leq s^{-\gamma} \abs{\beta} \sum_{j\in\N_n} \abs{V_{ij}} \rho(i,j)^\gamma.
  \end{equation*}

  \noindent Defining $c_2 = \sup_n \max_i \sum_{j\in\N_n} \abs{\beta} \abs{V_{ij}} \rho(i,j)^\gamma$ and $c = \max\{c_1,c_2\}+1 > 0$, the inequality in \autoref{aani} holds with $\psi(s) = c\,\min\{s^{-\gamma},1\}$ by construction. Furthermore, $c<\infty$ by \eqref{coani} and uniform boundedness of $\{\varepsilon_i\}_{i\in\N_n}$, and $\psi(s)$ satisfies $\sum_{s=1}^\infty s\,\psi(s) < \infty$ because $\gamma>2$.
\end{proof}

\noindent This result shows that, unlike the standard approach of imposing a $K$-neighborhood exposure mapping, \autoref{aani} can allow for richer forms of interference in which alters that are arbitrarily distant from the ego can interfere with the ego's response.

\begin{remark}[Literature]\label{moreleung0}
  \autoref{aani} and \autoref{pclifford} are spatial analogs of Assumptions 4 and 6 and Proposition 1 of \cite{leung2021causal} who studies interference mediated by an unweighted network, Bernoulli designs, and a different class of estimands defined by exposure mappings satisfying overlap. We study the global average treatment effect and cluster-randomized designs that induce overlap by introducing dependence in treatment assignments, and we further derive rate-optimal designs. These differences require an entirely distinct asymptotic theory.
\end{remark}

\subsection{Design and Estimator}\label{sdesign}

Much of the literature on interference considers designs in which units are individually randomized into treatment and control, either via Bernoulli or complete randomization. A common problem faced by such designs is limited overlap, meaning that some realizations of the exposure mapping occur with low probability. For example, suppose that \eqref{cspec} holds with exposure mapping $T_i = \sum_{j\in\N_n} \bm{1}\{\rho(i,j)\leq K\}D_j$, the number of treated units in $i$'s $K$-neighborhood. Then in a Bernoulli design, for large values of $K$, $\prob(T_i=0)$ is small, tending to zero with $K$ at an exponential rate. This is problematic for a Horvitz-Thompson estimator such as $n^{-1} \sum_{i\in\N_n} (p_i(t)^{-1} \bm{1}\{T_i=t\} - p_i(t')^{-1} \bm{1}\{T_i=t'\})Y_i$ where $p_i(t) = \prob(T_i=t)$ since its variance grows rapidly with $K$ if either $t$ or $t'$ is zero. Ugander et al.\ \cite{ugander2013graph} propose cluster-randomized designs, which reduce this problem by deliberately introducing dependence in treatment assignments across certain units. 

We consider the following class of such designs. We assign units to mutually exclusive clusters by partitioning the population region $Q(\bm{0},R_n)$ into $m_n \leq n$ equally sized squares, assuming for simplicity that $m_n \in \{4^s\colon s\in\mathbb{N}\}$. That is, to obtain increasingly more clusters, we first divide the population region into four squares, then divide each of these squares into four squares, and so on, as in \autoref{sequence}. Label the $m_n$ squares $Q_1, \dots, Q_{m_n}$, and call 
\begin{equation*}
  C_k = Q_k \medcap \N_n
\end{equation*}

\noindent cluster $k$. Then the number of units in each cluster is uniformly $O(n/m_n)$ under \autoref{aid}, and the radius of each cluster is 
\begin{equation*}
  r_n \equiv R_n/\sqrt{m_n} = \sqrt{n/m_n} + o(1),
\end{equation*}

\noindent which we assume is greater than 1. We also assume there are no units on the common boundaries of different squares, so the squares partition $\N_n$. 

\begin{figure}
  \centering
  \includegraphics[scale=0.5]{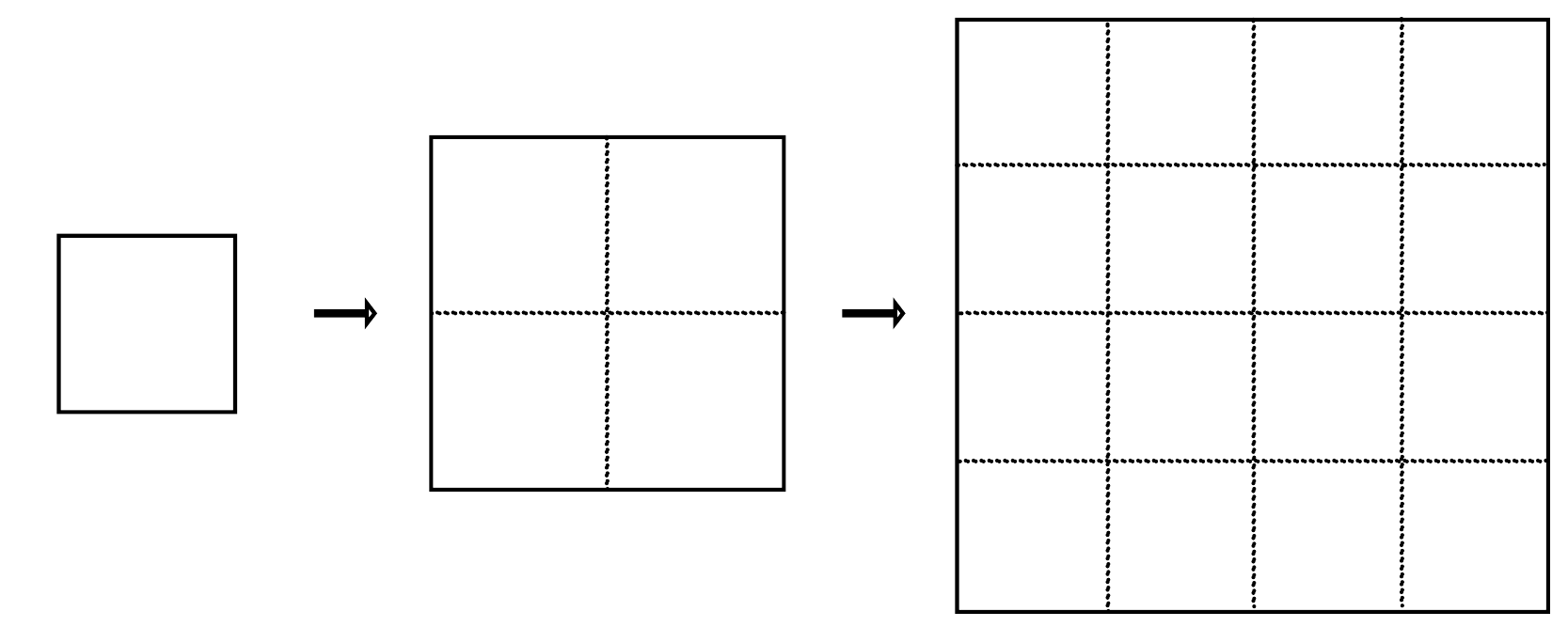}
  \caption{Sequence of population regions and clusters as $n\rightarrow\infty$.}
  \label{sequence}
\end{figure}

A {\em cluster-randomized design} first independently assigns each cluster to treatment and control with some probability 
\begin{equation*}
  p \in (0,1)
\end{equation*}

\noindent that is fixed with respect to $n$. Then within a treated (control) cluster $C_k$, all $i\in C_k$ are assigned $D_i=1$ ($D_i=0$). In order to emphasize that we use this design in later theorems, we state it as a separate assumption.

\begin{assump}[Design]\label{adesign}
  For any $n$, $\bm{D}$ is realized according to a cluster-randomized design with $m_n$ clusters constructed as above.
\end{assump}

Note that $m_n$ will be required to diverge with $n$ since a large number of clusters is needed for the estimator to concentrate. If $m_n$ is order $n$, then $r_n=O(1)$, so clusters are asymptotically bounded in size, the usual case studied in the literature, which includes unit-level Bernoulli randomization as a special case. If $m_n$ is of smaller order, then cluster sizes grow with $n$.

To construct the estimator, define the {\em neighborhood exposure indicator} 
\begin{equation*}
  T_{ti} = \prod_{j\in \N(i,\kappa_n)} \bm{1}\{D_j=t\} \quad\text{for}\quad t \in \{0,1\},\quad \kappa_n = r_n/2.
\end{equation*}

\noindent This is an indicator for whether $i$'s $\kappa_n$-neighborhood is entirely treated ($t=1$) or untreated ($t=0$). Unlike $K$-neighborhood exposure mappings, the radius $\kappa_n$ will be allowed to diverge. Let $p_{ti} = \E[T_{ti}]$. We study the Horvitz-Thompson estimator
\begin{equation*}
  \hat\theta = \frac{1}{n} \sum_{i\in\N_n} Z_i \quad\text{for}\quad Z_i = \left( \frac{T_{1i}}{p_{1i}} - \frac{T_{0i}}{p_{0i}} \right) Y_i.
\end{equation*}

\noindent Intuitively, $n^{-1} \sum_{i=1}^n Y_i T_{1i} p_{1i}^{-1}$ estimates $n^{-1} \sum_{i\in\N_n} Y_i(\bm{1}_n)$ using the outcomes of units whose neighbors within radius $\kappa_n$ are all treated. Since the radius depends on $m_n$ through $r_n$, $\hat\theta$ is a function of the number of clusters dictated by the design. \autoref{nbhds} depicts the relationship between the clusters and the $\kappa_n$-neighborhoods that determine exposure.

Since nontrivial designs will include both treated and untreated units, $\hat\theta$ is biased for the global average treatment effect. The choice of design can trade off the size of the bias against that of the variance. In particular, small choices of $m_n$ (few clusters, large radii) induce lower bias and higher variance. In \autoref{srmse}, we discuss nearly rate-optimal choices of $m_n$ for which the bias is asymptotically negligible.

\begin{remark}[Overlap]\label{roverlap}
  Under Bernoulli randomization, overlap needs to be imposed as a separate assumption for asymptotic inference. By overlap we mean the probability weights $p_{1i}$ and $p_{0i}$ are uniformly bounded away from zero and one, which imposes potentially strong restrictions on the types of exposure mappings the analyst can use, as previously illustrated. In our setup, however, overlap is directly satisfied because $p_{1i} = p^k$ and $p_{0i}=(1-p)^k$, where $k$ is the number of clusters that intersect $i$'s $\kappa_n$-neighborhood. Our choice of $\kappa_n$ implies $k\in [1,4]$ for all $i$, so overlap holds. 
\end{remark}

\begin{remark}[Neighborhood radius]
  Let $\bm{c}(C_k) \in \R^2$ be the centroid of cluster $C_k$. The choice of $\kappa_n = r_n/2$ ensures that, for any unit $i$ in the ``interior'' of a cluster in the sense that $i \in Q(\bm{c}(C_k),r_n/2)$, $i$'s $\kappa_n$-neighborhood also lies within that cluster, in which case the exposure probabilities are simply given by the cluster assignment probability: $p_{1i}=p$ and $p_{0i}=1-p$. If we had instead chosen, say, $\kappa_n=r_n$, then this would be true only for the centroid, while for the remaining units, $p_{1i}$ and $p_{0i}$ could be as small as $p^4$, which means less overlap and a more variable estimate. For the purposes of the asymptotic theory, the main requirement is that $\kappa_n$ has the same asymptotic order as $r_n$. If $\kappa_n$ were of smaller order, then results in \autoref{srmse} show that the bias of $\hat\theta$ could be non-negligible, whereas if $\kappa_n$ were of larger order, then $k$ in \autoref{roverlap} would grow with $n$, overlap would be limited, and $\var(\hat\theta)$ could be large.
\end{remark}

\section{Rate-Optimal Designs}\label{srmse}

We next derive the rate of convergence of $\hat\theta$ as a function of $n$, $m_n$, and $\psi(\cdot)$, which we use to obtain rate-optimal choices of $m_n$. Recall that designs are parameterized by $m_n$, which determines the number and sizes of clusters, and also that $\hat\theta$ depends on $m_n$ through the neighborhood exposure radius $\kappa_n$, so we will be optimizing over both the design and the radius that determines the estimator. 

\subsection{Rate of Convergence}

We first provide asymptotic upper bounds on the bias and variance of $\hat\theta$. For two sequences $\{a_n\}_{n\in\mathbb{N}}$ and $\{b_n\}_{n\in\mathbb{N}}$, we write $a_n \lesssim b_n$ to mean $a_n/b_n = O(1)$ and $a_n \gtrsim b_n$ to mean $b_n/a_n = O(1)$.

\begin{theorem}\label{tmse}
  Under Assumptions \ref{aboundY}--\ref{adesign}, $\abs{\E[\hat\theta] - \theta_n} \lesssim \psi(r_n/2)$ and $\var(\hat\theta) \lesssim m_n^{-1}$.
\end{theorem}
\begin{proof}
  First, we bound the bias. If $T_{1i}=1$, then all units in $Q(i,\kappa_n)$ are treated, so by \autoref{aani},
  \begin{equation*}
    \abs{\E[Y_iT_{1i}p_{1i}^{-1}] - Y_i(\bm{1}_n)} = \abs{\E[Y_i(\bm{D}) \mid T_{1i}=1] - Y_i(\bm{1}_n)} \leq \psi(\kappa_n).
  \end{equation*}

  \noindent The same argument applies to $\abs{\E[Y_iT_{0i}p_{0i}^{-1}] - Y_i(\bm{0}_n)}$, so combining these results, we obtain the rate for the bias.

  Next, we bound the variance. The following is an important object in our analysis and also for later constructing the variance estimator:
  \begin{equation}
    \Lambda_i = \{j\in\N_n\colon \exists \,k \,\text{ s.t. }\, C_k \cap \N(i,\kappa_n) \neq \emptyset,\, C_k \cap \N(j,\kappa_n) \neq \emptyset\}.
    \label{Lambdai}
  \end{equation}
  
  \noindent This is the set of units $j$ whose $\kappa_n$-neighborhoods intersect a cluster $C_k$ that also intersects $i$'s $\kappa_n$-neighborhood. We have
  \begin{equation}
    \var\left( \frac{1}{n} \sum_{i\in\N_n} Z_i \right) = \frac{1}{n^2} \sum_{i\in\N_n} \sum_{j \in \Lambda_i} \cov(Z_i,Z_j) + \frac{1}{n^2} \sum_{i\in\N_n} \sum_{j \not\in \Lambda_i} \cov(Z_i,Z_j) \equiv [A] + [B]. \label{AB}
  \end{equation}

  \noindent Note that $\Lambda_i$ contains units from at most 16 clusters (the worst case is when $Q(i,\kappa_n)$ intersects four clusters), and clusters contain uniformly $O(n/m_n)$ units by Lemma A.1 of \cite{jenish2009central}.
  By \autoref{aboundY} and \autoref{roverlap}, $\{Z_i\}_{i\in\N_n}$ is uniformly bounded, so
  \begin{equation*}
    [A] \lesssim \frac{1}{n^2} \cdot n \cdot \frac{n}{m_n} = \frac{1}{m_n}.
  \end{equation*}

  The difficult part of the argument is obtaining a tight enough rate for $[B]$, which requires control over the dependence between outcomes of units $i,j$ that are ``distant'' in the sense that $j\not\in\Lambda_i$. \autoref{Bbound} in \autoref{sproofs} proves that $[B] \lesssim (nm_n)^{-1/2}$. Therefore, $\var(\hat\theta) \lesssim m_n^{-1} + (nm_n)^{-1/2}$, which is at most $m_n^{-1}$ since $m_n \leq n$.
\end{proof}

Our second result provides asymptotic lower bounds.

\begin{theorem}\label{tsharp}
  Suppose $\psi(s) = \min\{s^{-\gamma},0.5\}$ for some $\gamma>2$. Under \autoref{adesign}, there exists a sequence of units and potential outcomes $\{\{Y_i(\cdot)\}_{i\in\mathcal{N}_n}\colon n\in\mathbb{N}\}$ satisfying Assumptions \ref{aboundY}--\ref{aani} such that $\abs{\E[\hat\theta] - \theta_n} \gtrsim \psi(1.5r_n)$ and $\var(\hat\theta) \gtrsim m_n^{-1}$.
\end{theorem}
\begin{proof}
  See supplemental material \citep{leung2022supplement}. 
\end{proof}

\noindent The result shows that we can construct potential outcomes satisfying the assumptions of \autoref{tmse} such that the bias is at least order $\psi(1.5r_n)$ and the variance at least $m_n^{-1}$. As discussed in \autoref{sintro}, existing work on cluster randomization under interference assumes clusters have asymptotically bounded size, which, in our setting, implies $r_n=O(1)$. \autoref{tsharp} implies that the bias of the Horvitz-Thompson estimator can then be bounded away from zero, showing that existing results strongly rely on the exposure mapping assumption to obtain unbiased estimates. In the absence of this assumption, it is necessary to consider designs in which cluster sizes are large to ensure the bias vanishes with $n$.

\subsection{Design Examples}\label{sdex}

\autoref{tmse} implies the mean-squared error of $\hat\theta$ is at most of order $\psi(r_n/2)^2 + m_n^{-1}$, and \autoref{tsharp} provides a similar asymptotic lower bound. Under either bound, the bias increases with $m_n$ while the variance decreases, so there exists a bias-variance trade-off in the choice of design. We next derive rates for $m_n$ that minimize or nearly minimize the upper bound under different assumptions on $\psi(\cdot)$. Based on these results, we make recommendations for practical implementation in the next subsection.

{\em Oracle design.} Suppose $\psi(s)$ is known to decay with $s$ at some rate $\phi(s)$. Then by definition of $r_n$, a rate-optimal design chooses $m_n$ to minimize $\phi( 0.5R_nm_n^{-1/2} )^2 + m_n^{-1}$. 

{\em Exposure mappings.} If we assume \eqref{cspec} holds for some $K$-neighborhood exposure mapping, then $\psi(s) = 0$ for all $s>K$. If $K$ is known, then by choosing $m_n = R_n^2(2K)^{-2}$, we have $\kappa_n = K$ and zero bias. In this case, clusters are asymptotically bounded in size, the estimator converges at rate $n^{-1/2}$, and both the design and estimator qualitatively coincide with those of \cite{ugander2013graph}. 

On the other hand, if $K$ is unknown, then for a nearly rate-optimal design, we can choose $\kappa_n$ to grow at a slow rate so that it eventually exceeds any fixed $K$. This may be achieved by choosing $m_n$ to grow slightly slower than $n$, say $n/\log(n)$. Then for large enough $n$, the bias is zero, and the rate of convergence is $\sqrt{\log(n)/n}$.

{\em Exponential decay.} Common specifications of the spatial weight matrix $\bm{W}$ in the Cliff-Ord model imply that $\psi(s)$ decays exponentially with $s$, for example, the row-normalized matrix
\begin{equation}
  W_{ij} = \frac{\bm{1}\{\rho(i,j) \leq 1\}}{\sum_{k\in\N_n} \bm{1}\{\rho(i,k) \leq 1\}}. \label{Wadj}
\end{equation}

\noindent If $\psi(s)$ is known to decay at some exponential rate but the exponent is unknown, then we may choose $m_n = n^{1-\epsilon}$ for any small $\epsilon>0$ for a nearly rate-optimal design, which yields a rate of convergence of $n^{-0.5(1-\epsilon)}$, close to an $n^{-1/2}$-rate. This shows that rates close to $n^{-1/2}$ are attainable in the absence of exposure mapping assumptions, despite targeting the global average treatment effect.

{\em Worst-case decay.} In practice, we may have little prior knowledge about $\psi(s)$. Recall that \autoref{aani} requires the rate of decay to be no slower than $s^{-2(1+\epsilon)}$ for $\epsilon>0$. As discussed in \autoref{rrate}, this is the slowest rate for spatial dependence that ensures a finite variance. For this rate, since $R_n$ is order $\sqrt{n}$, the bias is order $(n/m_n)^{-(1+\epsilon)}$. Without knowledge of $\epsilon$, we can settle for a nearly rate-optimal design by setting $\epsilon=0$ and choosing $m_n$ to minimize $(n/m_n)^{-2} + m_n^{-1}$, which yields $m_n = n^{2/3}$ and an $n^{-1/3}$-rate of convergence. Under this design, cluster sizes grow at the rate $r_n^2 = n^{1/3}$.

In the last three designs, the bias is $o(m_n^{-1/2})$, which is of smaller order than the variance. This makes the bias negligible from an asymptotic standpoint, but it would be useful to develop bias-reduction methods. We also reiterate that, while this analysis only provides rates, it is apparently by necessity at this level of generality. A finite-sample optimal design seems to require substantially more knowledge of the functional form of $\psi(\cdot)$.

\subsection{Practical Recommendations}\label{sprac}

The designs in the previous section rely on varying degrees of knowledge of $\psi(\cdot)$, the rate at which interference vanishes with distance. In practice, this is likely unknown, so we recommend operating under the worst-case rate of decay discussed in the previous subsection. The {\em default conservative choice} we recommend using in practice is the near-optimal rate described there, namely
\begin{equation}
  m_n = n^{2/3}. \label{dcc}
\end{equation}

To construct the clusters, we recommend partitioning space into $m_n$ clusters using a clustering algorithm, such as spectral clustering. A confidence interval (CI) for $\theta_n$ is given in \eqref{CI}. In \autoref{smc}, we explore in simulations the performance of the CI when clusters are constructed according to these recommendations.

Our large-sample theory assumes space is subdivided into evenly sized squares in order to avoid the difficult problem of optimizing over arbitrary shapes. However, since units are typically irregularly distributed in practice, division into equally sized squares may be inefficient, which is why we recommend the use of clustering algorithms. We suggest spectral clustering because it recovers, under weak conditions, low-conductance clusters \citep{peng2017partitioning}, and low conductance is the key property of clusters utilized in our proofs, as discussed in \autoref{sconclude}.

\section{Inference}\label{sinfer}

We next state results for asymptotic inference on $\theta_n$. Define $\sigma_n^2 = \var(\sqrt{m_n}\hat\theta)$.

\begin{assump}[Non-degeneracy]\label{anondeg}
  $\liminf_{n\rightarrow\infty} \sigma_n^2 > 0$.
\end{assump}

\noindent This is a standard condition and reasonable to impose in light of the lower bound on the variance derived in \autoref{tsharp}.

\begin{theorem}\label{tclt}
  Suppose $m_n\rightarrow\infty$ and $m_n=o(n)$. Under Assumptions \ref{aboundY}--\ref{anondeg},
  \begin{equation*}
    \sigma_n^{-1} \sqrt{m_n} (\hat\theta - \E[\hat\theta]) \dlimarrow \N(0,1).
  \end{equation*}
\end{theorem}
\begin{proof}
  See \autoref{sproofs}.
\end{proof}

\noindent The result centers $\hat\theta$ at its expectation, not the estimand $\theta_n$. However, designs discussed in \autoref{sdex} result in small bias, meaning $\abs{\E[\hat\theta] - \theta_n} = o(m_n^{-1/2})$, so we can replace $\E[\hat\theta]$ with $\theta_n$ on the left-hand side. Also note that the assumption $m_n=o(n)$ implies that cluster sizes grow with $n$. If instead $m_n$ were of order $n$, then $r_n=O(1)$, so by \autoref{tsharp}, we would additionally need to assume that there exists a $K$-neighborhood exposure mapping in the sense of \eqref{cspec} in order to guarantee that the bias vanishes at all. In this case, it is straightforward to establish a normal approximation using existing results.

\subsection{Proof Sketch}\label{ssketch}

To our knowledge, there is no off-the-shelf central limit theorem that we can apply to $\hat\theta$. Under \autoref{aani}, the outcomes appear to be near-epoch dependent on the input process $\{D_i\}_{i\in\N_n}$, but the treatments are cluster-dependent with growing cluster sizes, rather than $\alpha$-mixing, as required by \cite{jenish2012spatial}. To prove a central limit theorem, they split the average into two parts: its expectation conditional on the dependent input process $\{D_i\}_{i\in\N_n}$, and a remainder that they show is small. Rather than conditioning on all treatments, we find that the following unit-specific conditioning event is more useful for proving our result.

Let $\C_i$ be the cluster containing unit $i$, and $\F_i = \{D_j\colon j \in \C_i \cup \N(i,\kappa_n)\}$. Rewrite the estimator as
\begin{equation}
  \hat\theta = \frac{1}{n} \sum_{i\in\N_n} Z_i = \frac{1}{n} \sum_{i\in\N_n} \E[Z_i \mid \F_i] + \frac{1}{n} \sum_{i\in\N_n} (Z_i - \E[Z_i \mid \F_i]). \label{decomp}
\end{equation}

\noindent We first show that the last term is relatively small, $o_p(m_n^{-1/2})$ to be precise, which means that, on average, $Z_i$ is primarily determined by $\F_i$. The proof of this claim is somewhat complicated \cite[and different from that of][]{jenish2012spatial}, but it is similar to the argument showing $[B] \lesssim (nm_n)^{-1/2}$ in \eqref{AB}. To then establish a central limit theorem for $n^{-1} \sum_{i\in\N_n} \E[Z_i \mid \F_i]$, we observe that the dependence between ``observations'' $\{\E[Z_i \mid \F_i]\}_{i\in\N_n}$ is characterized by the following dependency graph $\bm{A}$, which, roughly speaking, links two units only if they are dependent. Recalling the definition of $\Lambda_i$ from \eqref{Lambdai}, we connect units $i,j$ in $\bm{A}$ if and only if $j \in \Lambda_i$ (or equivalently $i \in \Lambda_j$). Then $\bm{A}$ is indeed a dependency graph because, under \autoref{adesign}, $j \not\in\Lambda_i$ implies that the treatment assignments that determine $\F_i$ are independent of those that determine $\F_j$. The result follows from a central limit theorem for dependency graphs.

The proof highlights two sources of dependence. The first-order source is the first term on the right-hand side of \eqref{decomp}. Dependence in this average is due to cluster randomization, which induces correlation in the treatments determining $\F_i$ across $i$. The second-order source is the second term on the right-hand side of \eqref{decomp}. Dependence in this average is due to interference, which decays with distance due to \autoref{aani}. S\"{a}vje \cite{savje2021causal} derives a similar decomposition in a different context with misspecified exposure mappings. The previous arguments show that the second-order source of dependence is small relative to the first-order source because, with large clusters, dependence induced by cluster randomization dominates dependence induced by interference. This is generally untrue with small clusters. 

\subsection{Variance Estimator}\label{svarest}

The proof sketch suggests that, to estimate $\sigma_n^2$, it suffices to account for dependence induced by cluster randomization. Define $A_{ij} = \bm{1}\{j \in \Lambda_i\}$, where $\Lambda_i$ is defined in \eqref{Lambdai}, and note that $A_{ii}=1$ and $A_{ij}=A_{ji}$. Let $\bar{Z} = n^{-1}\sum_{i\in\N_n} Z_i$, which is equivalent to $\hat\theta$. Our proposed variance estimator is
\begin{equation}
  \hat\sigma^2 = \frac{m_n}{n^2} \sum_{i\in\N_n} \sum_{j\in\N_n} (Z_i - \bar{Z}) (Z_j - \bar{Z}) A_{ij}. \label{hatsigma}
\end{equation}

\begin{theorem}\label{tvar}
  Suppose $m_n\rightarrow\infty$ and $m_n=o(n)$. Under Assumptions \ref{aboundY}--\ref{anondeg}, $(\hat\sigma^2-\mathcal{R}_n)/\sigma_n^2 \plimarrow 1$, where
  \begin{equation*}
    \mathcal{R}_n = \frac{m_n}{n} \hat{\tilde{\sigma}}^2 \quad\text{and}\quad \hat{\tilde{\sigma}}^2 = \frac{1}{n} \sum_{i\in\N_n} \sum_{j\in\N_n} (\E[Z_i] - \E[\bar{Z}]) (\E[Z_j] - \E[\bar{Z}]) A_{ij}.
  \end{equation*}
\end{theorem}
\begin{proof}
  See supplemental material \citep{leung2022supplement}.
\end{proof}

\noindent The bias term $\mathcal{R}_n$ is typically nonzero due to the unit-level heterogeneity. That is, $\abs{\E[Z_i] - \E[\bar{Z}]}$ does not approach zero asymptotically, except in the special case of homogeneous treatment effects where $Y_i(\bm{1}_n)-Y_i(\bm{0}_n)$ does not vary across $i$. In the no-interference setting, it is well-known that the variance of the difference-in-means estimator is biased for the same reason and that consistent estimation of the variance is impossible. However, due to the term $m_n/n = o(1)$ in $\mathcal{R}_n$, we will argue that typically $\mathcal{R}_n = o_p(1)$, meaning that $\hat\sigma^2$ is asymptotically exact.

Let us first compare $\mathcal{R}_n$ to its formulation under no interference. In this case, $Y_i(\bm{D}) = Y_i(D_i)$, and we replace $T_{1i}$ with $D_i$ and $T_{0i}$ with $1-D_i$ to estimate the usual average treatment effect. Furthermore, we set $A_{ij}=0$ for all $i\neq j$ because units are independent and set $m_n=n$ since there is no longer a need to cluster units. With these changes, $Z_i = (D_i/p - (1-D_i)/(1-p))Y_i$, and 
\begin{equation}
  \mathcal{R}_n = \frac{1}{n} \sum_{i\in\N_n} (\tau_i - \bar{\tau})^2 \label{nointer}
\end{equation}

\noindent for $\tau_i = Y_i(1)-Y_i(0)$ and $\bar{\tau} = n^{-1} \sum_{i\in\N_n} \tau_i$. This is the well-known expression for the bias in the absence of interference \citep[e.g.][Theorem 6.2]{imbens_causal_2015}. 

In our setting, we have additional ``covariance'' terms included in $\mathcal{R}_n$ due to the non-zero off-diagonals of the dependency graph $A_{ij}$. These would be problematic if they were negative and larger in magnitude than the main variance terms since that would make $\hat\sigma^2$ anti-conservative. We show that this occurs with small probability, and in fact, that $\mathcal{R}_n$ is $o_p(1)$. Observe that $m_n/n = o(1)$ and $\hat{\tilde{\sigma}}^2$ has the form of a HAC (heteroskedasticity and autocorrelation consistent) variance estimator \citep{andrews1991heteroskedasticity,conley1999gmm}. Hence, under conventional regularity conditions, $\hat{\tilde{\sigma}}^2$ is consistent for a variance term $\tilde\sigma^2 \geq 0$, in which case $\mathcal{R}_n$ is non-negative in large samples, and furthermore, $o_p(1)$. To formalize this intuition, we need to specify conditions on the superpopulation from which potential outcomes are drawn. In \autoref{sbias}, we show that, if potential outcomes are $\alpha$-mixing, then $\hat{\tilde{\sigma}}^2$ is asymptotically unbiased for $\tilde{\sigma}^2 = \var(n^{-1/2} \sum_{i=1}^n \E[Z_i \mid \{Y_i(\bm{d})\}_{\bm{d}\in\{0,1\}^n}])$, and furthermore, $\var(\hat{\tilde{\sigma}}^2) = O(n^2/m_n^3)$. Consequently, $\var(\mathcal{R}_n) = O(m_n^{-1})$ due to the $m_n/n$ term in its expression.

\begin{remark}[Confidence interval]
  As previously discussed, the bias of $\hat\theta_n$ is $o(m_n^{-1/2})$ for the near-optimal designs discussed at the end of \autoref{srmse}. Thus, for such designs, the preceding discussion justifies the use of 
  \begin{equation}
    \hat\theta \pm 1.96 \cdot \hat\sigma m_n^{-1/2} \label{CI}
  \end{equation} 

  \noindent as an asymptotic 95-percent CI for $\theta_n$, where $\hat\sigma^2$ is defined in \eqref{hatsigma}.
\end{remark}

\begin{remark}[Literature]\label{moreleung}
  Leung \cite{leung2021causal} proves a result similar to \autoref{tvar} but for a different variance estimator under a different design and variety of interference. Due to the lack of an analogous $m_n/n$ term, in his setting, weak dependence conditions would only ensure $\mathcal{R}_n \plimarrow \tilde\sigma^2 \geq 0$, in which case the estimator would be asymptotically conservative, whereas ours is asymptotically exact. He does not provide a formal result on the limit of $\mathcal{R}_n$.
\end{remark}

\section{Monte Carlo}\label{smc}

We next present results from a simulation study illustrating the quality of the normal approximation in \autoref{tclt} and coverage of the CI \eqref{CI} when constructing clusters using spectral clustering. To generate spatial locations, let $\{\tilde\rho_i\}_{i\in\N_n}$ be i.i.d.\ draws from $\mathcal{U}([-1,1]^2)$. Unit locations in $\R^2$ are given by $\{\rho_i\}_{i\in\N_n}$ for $\rho_i = R_n\tilde\rho_i$ with $R_n=\sqrt{n}$. We let $\rho(i,j) = \norm{\rho_i - \rho_j}$ where $\norm{\cdot}$ is the Euclidean norm.

We set the number of clusters according to \eqref{dcc}, rounded to the nearest integer, which corresponds to the near-optimal design under the worst-case decay discussed in \autoref{sdex}. To construct clusters, we apply spectral clustering to $\{\rho_i\}_{i\in\N_n}$ with the standard Gaussian affinity matrix whose $ij$th entry is $\text{exp}\{-\rho(i,j)^2\}$. Clusters are randomized into treatment with probability $p=0.5$. \autoref{speclusts} displays the clusters and treatment assignments for a typical simulation draw.

\begin{figure}
  \centering
  \includegraphics[scale=0.4]{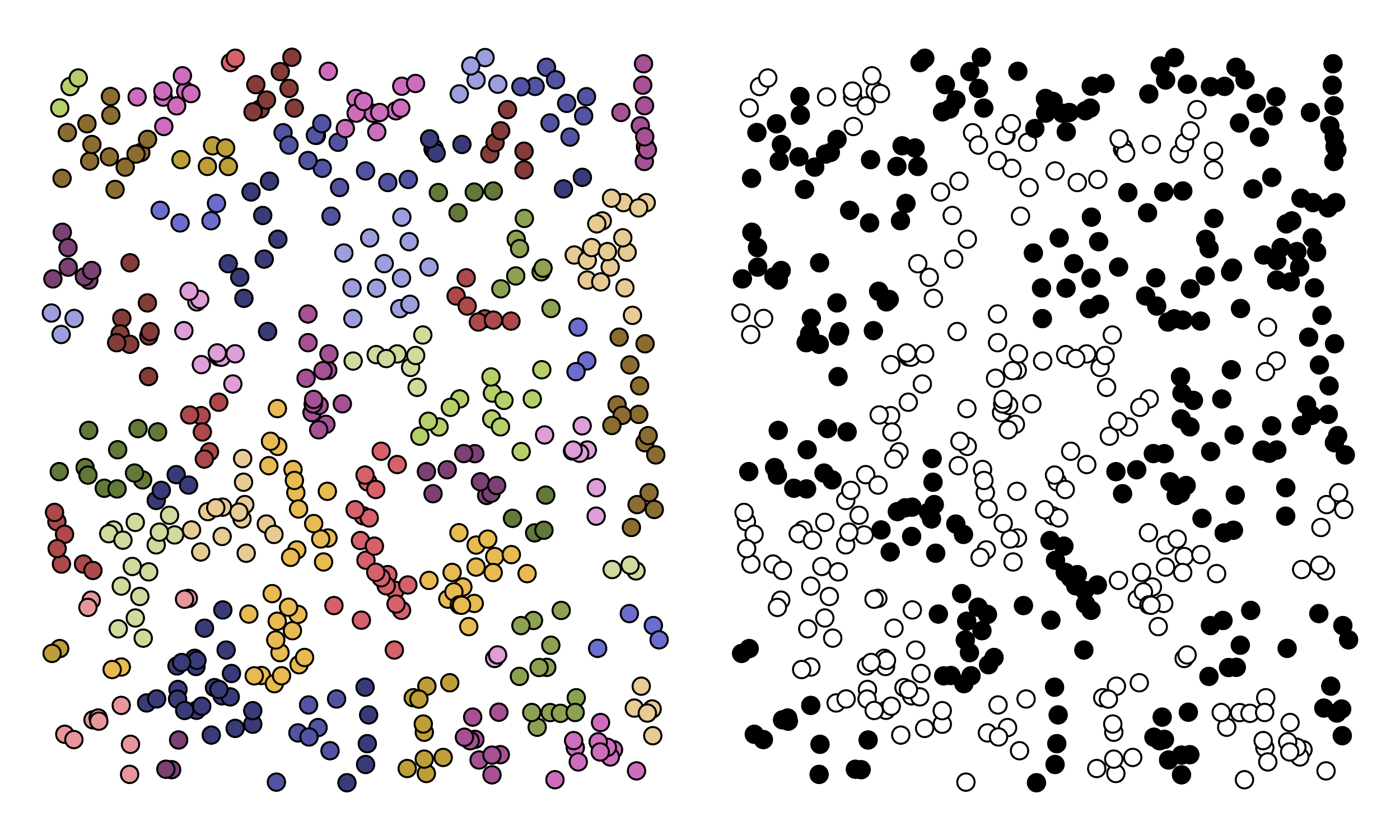}
  \caption{The left figure colors units by cluster memberships obtained from spectral clustering (some colors are reused for different clusters). The right figure colors units by treatment assignment.}
  \label{speclusts}
\end{figure}

We generate outcomes from three different models. Let $\{\varepsilon_i\}_{i\in\N_n} \stackrel{iid}\sim \mathcal{N}(0,1)$ be drawn independently of the other primitives. The first model is Cliff-Ord:
\begin{equation*}
  Y_i = \alpha + \lambda \sum_{j\in\mathcal{N}_n} W_{ij}Y_j + \delta \sum_{j\in\mathcal{N}_n} W_{ij}D_j + D_i\beta + \varepsilon_i
\end{equation*}

\noindent with $(\alpha,\lambda,\delta,\beta) = (-1,0.8,1,1)$ and spatial weight matrix given by the row-normalized adjacency matrix \eqref{Wadj}. As discussed in \autoref{srmse}, this model features exponentially decaying $\psi(s)$, in fact of order $\lambda^s$ \citep[][Proposition 1]{leung2021causal}.  

We construct the second and third models to explore how our methods break down when \autoref{aani} is violated or close to violated. For this purpose, we use the ``moving average'' model \eqref{clifford} with $(\alpha,\beta) = (-1,1)$ and $V_{ij} = \rho(i,j)^{-\eta}$ for $\eta=4,5$ for the two respective models, so that $\psi(s)$ decays at a polynomial rate. Notably, the choice of $\eta=4$ implies that the rate of decay is slow enough that \autoref{aani} can fail to hold. This is because
\begin{equation*}
  \eqref{coani} = \sum_{s=1}^\infty \sum_{j\in\N_n} \bm{1}\{\rho(i,j)\in [s-1,s)\} \rho(i,j)^{\gamma-4} \leq c \sum_{s=0}^\infty s^{\gamma-3}
\end{equation*}

\noindent for some $c>0$ by Lemma A.1(iii) of \cite{jenish2009central}. The right-hand side does not converge for some $\gamma>2$, as required by \autoref{pclifford}. On the other hand, the choice of $\eta=5$ is large enough for \autoref{aani} to be satisfied since we now replace the 3 on the right-hand side of the previous display with 4. However, in smaller samples, $\eta=4$ or 5 may not be substantially different, so our methods may still break down from the assumption being ``close to'' violated.

\autoref{tsims} displays the results of 5000 simulation draws. Row ``Bias$(\hat\theta)$'' displays $\abs{\E[\hat\theta - \theta_n]}$, estimated by taking the average over the draws, while ``Var$(\hat\theta)$'' is the variance of $\hat\theta$ across the draws. The next rows display the coverage of three different confidence intervals. ``Our CI'' corresponds to the empirical coverage of \eqref{CI}. ``Naive CI'' corresponds to \eqref{CI} but replaces $\hat\sigma m_n^{-1/2}$ with the i.i.d.\ standard error, so the extent to which its coverage deviates from 95 percent illustrates the degree of spatial dependence. ``Oracle CI'' corresponds to \eqref{CI} but replaces $\hat\sigma m_n^{-1/2}$ with the ``oracle'' SE, which is the standard deviation of $\hat\theta$ across the draws. Note that the oracle SE approximates $\sigma_n^2 + \mathcal{R}_n$ because the variance is taken over the randomness of the design as well as of the potential outcomes. Lastly, ``SE'' displays our standard error $\hat\sigma m_n^{-1/2}$. 

\begin{table}[ht]
\centering
\caption{Simulation Results}
\begin{threeparttable}
\begin{tabular}{lrrrrrrrrr}
\toprule
{} & \multicolumn{3}{c}{Moving Avg, $\eta=4$} & \multicolumn{3}{c}{Moving Avg, $\eta=5$} & \multicolumn{3}{c}{Cliff-Ord} \\
$n$ &                       250  &   500  &    1000 &                       250  &   500  &    1000 &      250  &   500  &    1000 \\
\midrule
$m_n$      &                     40 & 63 & 100 &  40 & 63 & 100 &  40 & 63 & 100 \\
Our CI   &                  0.909 &  0.925 &   0.924 &                  0.922 &  0.937 &   0.940 &     0.979 &  0.983 &   0.982 \\
Naive CI  &                  0.530 &  0.489 &   0.447 &                  0.575 &  0.539 &   0.494 &     0.918 &  0.913 &   0.916 \\
Oracle CI &                  0.943 &  0.943 &   0.936 &                  0.950 &  0.952 &   0.949 &     0.983 &  0.982 &   0.982 \\
Bias$(\hat\theta)$       &                  0.108 &  0.093 &   0.083 &                  0.033 &  0.027 &   0.024 &     0.009 &  0.004 &   0.003 \\
Var$(\hat\theta)$   &                  0.143 &  0.087 &   0.054 &                  0.108 &  0.065 &   0.039 &     0.341 &  0.177 &   0.087 \\
SE         &                  0.364 &  0.292 &   0.232 &                  0.317 &  0.252 &   0.199 &     0.601 &  0.432 &   0.309 \\
$\hat\theta$   &                  1.432 &  1.467 &   1.492 &                  1.289 &  1.307 &   1.319 &     5.804 &  5.822 &   5.851 \\
\bottomrule
\end{tabular}
\begin{tablenotes}[para,flushleft]
  \footnotesize 5k simulations. The ``CI'' rows show the empirical coverage of 95\% CIs. ``Naive'' and ``Oracle'' respectively correspond to i.i.d.\ and true standard errors.
\end{tablenotes}
\end{threeparttable}
\label{tsims}
\end{table}

There are at most 100 clusters in all designs, and the rate of convergence is quite slow at $n^{-1/3}$ for our choice of $m_n$. Nonetheless, across all designs, the coverage of the oracle CI is close to 95 percent or above, which illustrates the quality of the normal approximation. For the Cliff-Ord model, our CI attains at least 95 percent coverage even for small sample sizes, despite $m_n$ being chosen suboptimally for the worst-case decay. For the moving average model with $\eta=5$, we see some under-coverage in smaller samples due to the larger bias, which is unsurprising from the above discussion, but coverage is close to the nominal level for larger $n$. The results for $\eta=4$, as expected, are worse since it is deliberately constructed to violate our main assumption. Once again, our CI exhibits under-coverage due to the larger bias, but coverage improves and bias decreases as $n$ grows. 

\section{Conclusion}\label{sconclude}

This paper studies the design of cluster-randomized experiments targeting the global average treatment effect under spatial interference. Each design is characterized by a parameter $m_n$ that determines the number and sizes of clusters. We propose a Horvitz-Thompson estimator that compares units with different neighborhood exposures to treatment, where the neighborhood radius is of the same order as clusters' sizes given by the design. We asymptotically bound the estimator's bias and variance as a function of $m_n$ and the degree of interference and derive rate-optimal choices of $m_n$. Our lower asymptotic bound shows that designs using small clusters (those with asymptotically bounded sizes) generally result in a non-negligible asymptotic bias. On the other hand, constructing large clusters reduces the total number of clusters, resulting in a bias-variance trade-off that we seek to optimize in terms of rates through the choice of design. 

In the worst case where the degree of interference is substantial, the estimator has an $n^{-1/3}$-rate of convergence under a nearly rate-optimal design, whereas in the best case where interference is characterized by a $K$-neighborhood exposure mapping, the rate is $n^{-1/2}$ under a rate-optimal design. We derive the asymptotic distribution of the estimator and provide an estimate of the variance. 

Important areas for future research include data-driven choices of $m_n$ and $\kappa_n$ and methods to reduce the bias of the estimator. However, a rigorous theory appears to require more substantive restrictions on interference than what we impose. 

Our results focus on the canonical case of spatial data in $\R^2$. We conjecture that they can be extended to $\R^d$ for $d>2$ because our proofs fundamentally rely on the following key property of Euclidean space, which is true for any dimension: it is always possible to construct many clusters with low {\em conductance}, or boundary-to-volume ratio, for example by partitioning space into hypercubes or by spectral clustering \citep{leung2021network}. This appears in our proofs through the use of Lemma A.1 of \cite{jenish2009central}, which, together with \autoref{aani}, is crucial to establish that spatially distant units have small covariance, despite dependence induced by cluster randomization and interference. In this sense, the technical idea behind this paper is to exploit a useful property of Euclidean space -- the existence of many low-conductance clusters -- to show that cluster-randomized designs may be fruitfully applied to the problem of spatial interference.

The story for network interference appears to be different. Existing cluster-randomized designs have theoretical guarantees under exposure mapping assumptions, but it is an open question whether such designs work under weaker restrictions on interference such as \autoref{aani}. In order to directly apply our idea in the previous paragraph, the network must possess many low-conductance clusters across which we can randomize. Unfortunately, this is a strong requirement in practice because, as discussed in \cite{leung2021network}, not only do some networks not possess multiple low-conductance clusters, but, of those that do, some apparently possess only a small number of such clusters. Because network ``space'' differs from Euclidean space in this fundamental aspect, under network interference, clusters can be strongly dependent in the absence of exposure mapping assumptions.


\begin{appendix}
\numberwithin{equation}{section}

\section{Bias of the Variance Estimator}\label{sbias}

Characterizing the asymptotic behavior of $\mathcal{R}_n$ requires conditions on the superpopulation from which units are drawn. In this section, we assume potential outcomes are random and constitute a weakly dependent spatial process (independent of $\bm{D}$). Accordingly, we rewrite $\E[Z_i]$ in \autoref{tvar} as
\begin{equation*}
  \tilde{Z}_i \equiv \E[Z_i \mid \{Y_i(\bm{d})\}_{\bm{d}\in\{0,1\}^n}].
\end{equation*}

\noindent We require the spatial process $\{\tilde{Z}_i\}_{i\in\N_n}$ to be $\alpha$-mixing, which is a standard concept of weak spatial dependence. The results we use in fact apply to the weaker concept of near-epoch dependence, but we focus on mixing since it requires less exposition.

\begin{secdefinition}
  Let $(\Omega, \F, \prob)$ be the probability space, $\mathcal{A},\mathcal{B}$ be sub-$\sigma$-algebras of $\F$, and
  \begin{equation*}
    \alpha(\mathcal{A},\mathcal{B}) = \sup\{ \abs{\prob(A \cap B) - \prob(A) \prob(B)}; A \in \mathcal{A}, B \in \mathcal{B} \}.
  \end{equation*}

  \noindent For $U,V \subseteq \N_n$, let $\alpha_n(U,V) = \alpha(\sigma_n(U), \sigma_n(V))$, where $\sigma_n(U)$ is $\sigma$-algebra generated by $\{\tilde{Z}_i\}_{i\in U}$. The {\em $\alpha$-mixing coefficient} of $\{\tilde{Z}_i\}_{i\in\N_n}$ is
  \begin{equation*}
    \bar{\alpha}(u,v,r) = \sup_n \sup_{U,V} \{\alpha_n(U,V); \abs{U}\leq u, \abs{V}\leq v, \rho(U,V) \geq r\} 
  \end{equation*}

  \noindent for $u,v\in\mathbb{N}$, $r\in\R_+$, and $\rho(U,V) = \min\{\rho(i,j)\colon i \in U, j \in V\}$.
\end{secdefinition}

\noindent That is, for any two sets of units $U,V \subseteq \N_n$ with respective sizes $u,v$ such that the minimum distance between $U,V$ is at least $r$, $\bar{\alpha}(u,v,r)$ bounds their dependence with respect to observations $\{\tilde{Z}_i\}_{i\in\N_n}$. 

\begin{secexample}
  Suppose that, for any $n$ and $i\in\N_n$, there exists a function $f(\cdot)$ such that $Y_i(\bm{d}) = f(\varepsilon_i, \bm{d})$. If the unobserved heterogeneity $\{\varepsilon_i\}_{i\in\N_n}$ is $\alpha$-mixing and $\tilde{Z}_i$ is a Borel-measurable function of $\varepsilon_i$ (a mild requirement since treatments are independent of potential outcomes), then $\{\tilde{Z}_i\}_{i\in\N_n}$ is $\alpha$-mixing.
\end{secexample}

\begin{secexample}
  Generalizing the previous example, suppose $Y_i(\bm{d}) = f(d_i, d_{-i}, \varepsilon_i, \varepsilon_{-i})$, where $\bm{d} = (d_i, d_{-i})$ and $\varepsilon_{-i}$ is similarly defined. Under some conditions on $f(\cdot)$, one can ensure that $\{\tilde{Z}_i\}_{i\in\N_n}$ is near-epoch dependent on the input $\{\varepsilon_i\}_{i\in\N_n}$ \citep[e.g.][Proposition 1]{jenish2012spatial}. However, we only focus on mixing.
\end{secexample}

\noindent We next discuss the intuition behind our main result. Let $\bar{\tilde{Z}} = n^{-1} \sum_{i=1}^n \tilde{Z}_i$, so that 
\begin{equation*}
  \mathcal{R}_n = \frac{m_n}{n} \hat{\tilde{\sigma}}^2, \quad\text{where}\quad \hat{\tilde{\sigma}}^2 = \frac{1}{n} \sum_{i\in\N_n} \sum_{j\in\N_n} (\tilde{Z}_i - \E[\bar{\tilde{Z}}]) (\tilde{Z}_i - \E[\bar{\tilde{Z}}]) A_{ij}.
\end{equation*}

\noindent Observe that $m_n/n=o(1)$, and $\hat{\tilde{\sigma}}^2$ is essentially a HAC variance estimator with ``kernel'' $A_{ij}$. More precisely, $A_{ij}$ is sandwiched between two uniform kernels:
\begin{equation}
  \bm{1}\{\rho(i,j) \leq \kappa_n\} \leq A_{ij} \leq \bm{1}\{\rho(i,j) \leq 2r_n+\kappa_n\}. \label{Asw}
\end{equation}

\noindent This is a consequence of the construction of clusters. The lower bound holds because $\Lambda_i$ must include $i$'s $\kappa_n$-neighborhood. The upper bound is achieved if $i$ is located at a corner shared by four clusters, and each such cluster intersects some $\kappa_n$-neighborhood that is maximally distant from the cluster. Notably, the bandwidths of the two kernels are of the same asymptotic order (recalling that $\kappa_n$ has the same order as $r_n$). Hence, $\hat{\tilde{\sigma}}^2$ should behave as a HAC variance estimator. This has two implications. 
\begin{enumerate}
  \item $\hat{\tilde{\sigma}}^2$ should be consistent for a non-negative variance term under standard regularity conditions, so $\mathcal{R}_n = o_p(1)$. 

  \item If, in the formula for $\hat\sigma^2$, we replace $A_{ij}$ with the uniform kernel in the upper bound \eqref{Asw}, then we would have a positive semidefinite HAC estimator, implying $\mathcal{R}_n \geq 0$ a.s.\ \citep[][\S3.3.1]{andrews1991heteroskedasticity,conley1999gmm}. Then in the finite-population setting of our main results, $\hat\sigma^2$ would be at worst asymptotically conservative.

    We choose not use a spatial HAC for $\hat\sigma^2$ because they have a reputation for being anti-conservative in smaller samples \citep[see references in e.g.][]{leung2021network}. In our estimator, $A_{ij}$ functions as a sort of heterogeneous bandwidth determined by the design that allows different units to have different neighborhood radii in the variance estimator, whereas HAC kernels imply homogeneous radii determined by the bandwidth. The hope is that heterogeneous radii could translate to better finite-sample properties since they directly capture the first-order dependency neighborhood.
\end{enumerate}

We next state regularity conditions taken from \cite{jenish2016spatial}, which we use to apply her Theorem 4 on consistency of HAC estimators.

\begin{secassump}\label{amix}
  The mixing coefficient satisfies $\bar{\alpha}(u,v,r) \leq (u+v)^\varsigma \hat{\alpha}(r)$ for $\varsigma \geq 0$ and $\sum_{r=1}^\infty r^{2(\varsigma+1)-1} \hat\alpha(r) < \infty$.
\end{secassump}

\noindent This is Assumption 2 of \cite{jenish2016spatial}. The substance of the condition is the requirement that the mixing coefficient decays at a sufficiently fast rate with distance $r$. For $\varsigma > 0$, the rate requirement is stronger than what we require of $\psi(\cdot)$ in \autoref{aani}.

\begin{secassump}\label{amoms}
  (a) $\E[\tilde{Z}_i]=\E[\tilde{Z}_j]$ for all $i,j$. (b) $\var(n^{-1/2} \sum_{i=1}^n \tilde{Z}_i) \plimarrow \tilde\sigma^2 < \infty$.
\end{secassump}

\noindent This is Assumption 7(a) of \cite{jenish2016spatial}. Part (a) is a standard mean-homogeneity condition required for HAC to be asymptotically unbiased. Such a requirement is untenable in the finite population setting of \autoref{tvar} because $\tilde{Z}_i$ is a function of $i$'s potential outcomes which are generally heterogeneous across units. Heterogeneity is responsible for the appearance of the bias $\mathcal{R}_n$. However, in the superpopulation setting of this section, the assumption is much more tenable since we integrate over the randomness of the potential outcomes. The condition then requires that the mean be invariant to unit labels. The finiteness requirement in part (b) can be proven as a consequence of \autoref{amix} and moment conditions, so the only substance of the assumption is the existence of a limit.

\begin{sectheorem}\label{aR_n}
  Under Assumptions \ref{aboundY}, \ref{aid}, \ref{amix}, and \ref{amoms}, if $m_n \rightarrow \infty$, then (a) $\E[\hat{\tilde{\sigma}}^2] \rightarrow \tilde\sigma^2$, and (b) $\var(\hat{\tilde{\sigma}}) = O(n^2 / m_n^3)$.
\end{sectheorem}

We next discuss the implications of the theorem for $\mathcal{R}_n$ (also see the discussion in \autoref{svarest}) and then conclude with the proof. The designs discussed at the end of \autoref{srmse}, other than under the worst-case decay, choose $m_n$ to be of substantially higher order than $n^{2/3}$. In this case, \autoref{aR_n} yields
\begin{equation*}
  \mathcal{R}_n = \underbrace{\frac{m_n}{n}}_{o(1)} \underbrace{\hat{\tilde{\sigma}}^2}_{\plimarrow \tilde\sigma^2 \geq 0} \plimarrow 0.
\end{equation*}

\noindent For the worst-case decay, $m_n = n^{2/3}$, in which case $\hat{\tilde{\sigma}}^2$ remains asymptotically unbiased for $\tilde\sigma^2 \geq 0$, and we still have $\var(\mathcal{R}_n) = O(m_n^{-1})$, so again $\mathcal{R}_n = o_p(1)$.

\begin{proof}[Proof of \autoref{aR_n}]
  We apply Theorem 4 of \cite{jenish2016spatial}. Note that our setting is a simple mean estimation problem, which is a special case of her semiparametric model $Y_{1in} = h(Y_{2in},\theta_0) + g(X_{in}) + U_{in}$ with $h(Y_{2in},\theta_0) + g(X_{in})=0$ and $U_{in}$ equal to our $\tilde{Z}_i - \E[\bar{\tilde{Z}}] = \tilde{Z}_i - \E[\tilde{Z}_i]$ (by \autoref{amoms}(a)). The moments $m_{\varepsilon_n}(W_{in},\hat\theta,\hat\tau)$ in her formula (13) for the HAC estimator corresponds to our $\tilde{Z}_i - \E[\tilde{Z}_i]$. Other than Assumption 10, her assumptions are either satisfied (increasing domain corresponds to our \autoref{aid} and the moment conditions hold by our \autoref{aboundY} and \autoref{roverlap}) or are irrelevant in our setting. 

  Assumption 10 concerns the properties of the kernel function. For context, note that if, hypothetically, we replaced $A_{ij}$ in $\hat{\tilde{\sigma}}^2$ with its upper bound in \eqref{Asw}, then the kernel $K(\cdot)$ and bandwidth $\beta_n$ in \cite{jenish2016spatial} would correspond in our setting to the uniform kernel $\bm{1}\{\cdot \leq 1\}$ and $2r_n+\kappa_n$, respectively, so that $K(x/\beta_n)$ in Jenish's notation would correspond to our $\bm{1}\{x \leq 2r_n+\kappa_n\}$. 

  Now, because $A_{ij}$ is only bounded by kernel functions but cannot be written as one, we cannot directly verify Assumption 10. However, inspection of the proof reveals that the assumption is used as follows. First, to derive bounds on the variance of the HAC estimator (Step 1 of the proof), uniform boundedness of the kernel is used, but this is trivially satisfied by $A_{ij} \in \{0,1\}$. Second, to derive bounds on the bias (Step 2 of the proof), Assumption 10 is used to show that the term $a_{r,n} = \argmax_{r\leq x\leq r+1} \abs{K(x/\beta_n)-1}\rightarrow 0$ as $n\rightarrow\infty$ for any $r>0$. In our case, $a_{r,n}$ corresponds to $\argmax_{i,j\in\N_n\colon \rho(i,j) \in [r,r+1]} \abs{A_{ij}-1}$. But this has the desired property; it is in fact exactly zero for $n$ sufficiently large due to \eqref{Asw}. 

  Hence, the conclusions of the proof of Jenish's Theorem 4 apply to $\hat{\tilde{\sigma}}^2$, which we now apply to prove our claims. Part (a) of our theorem follows from Step 2 of her proof. Next, in Step 1 of her proof, $H_{2n}, H_{3n}, H_{4n}=0$ in our setting because the data is $\alpha$-mixing rather than near-epoch dependent. Accordingly, the variance bound on $H_{1n}$ in that step implies that the variance of the HAC estimator is $O(n^{-1} \beta_n^{3d})$ where $\beta_n$ is the bandwidth and $d$ is the dimension of the space. In our case, by \eqref{Asw}, the bandwidth corresponds to $\beta_n = 2r_n+\kappa_n = O( \sqrt{n/m_n} )$, so $n^{-1} \beta_n^{3d} = O(n^2 / m_n^3)$, and part (b) of our theorem follows. 
\end{proof}

\section{Proofs}\label{sproofs}

The proofs use the following definitions. Let $\bm{c}(C_k) \in \R^2$ be the centroid of cluster $C_k$. For $s \leq r_n-1$, define
\begin{equation}
  J(s, C_k) = \big\{j \in C_k \colon \rho(\bm{c}(C_k),j)\in [r_n-s-1,r_n-s) \big\}. \label{Nnt}
\end{equation}

\noindent For $s=0$, this is the ``boundary'' of $C_k$, and as we increase $s$, $J(s, C_k)$ moves through contour sets within $C_k$ that are increasingly further from the boundary. Also, for any two sets $S,T \subset \R^2$, let $\rho(S,T) = \min\{\rho(i,j)\colon i\in S, j\in T\}$. 

The proofs make use of the following facts, which are a consequence of Lemma A.1 of \cite{jenish2009central} and use \autoref{aid}. Given that $C_k$ has radius $r_n$, $\abs{C_k} \leq c'r_n^2$ and $\abs{J(0,C_k)} \leq c'r_n$ for some universal constant $c'>0$ that does not depend on $n$ or $k$. Also, since $J(s, C_k)$ is the boundary of a ball of radius $r_n-s$, $\abs{J(s, C_k)} \leq c'(r_n-s)$.

\begin{seclemma}\label{Bbound}
  Recall the definition of $[B]$ from \eqref{AB}. Under the assumptions of \autoref{tmse}, $[B] \lesssim (nm_n)^{-1/2}$.
\end{seclemma}
\begin{proof}
  {\em Step 1.} We first establish covariance bounds. Fix $i,j$ such that $j \not\in \Lambda_i$, the latter defined in \eqref{Lambdai}, and set $s=\rho(i,j)$. Trivially,
  \begin{equation*}
    \abs{\cov(Z_i,Z_j)} = \abs{\cov(Z_i,Z_j)} \bm{1}\{s \leq 4r_n\} + \abs{\cov(Z_i,Z_j)} \bm{1}\{s > 4r_n\}.
  \end{equation*}

  \noindent First consider the case $s \leq 4r_n$. Let $\C_i$ be the cluster containing unit $i$, $\F_i(r) = \{D_j\colon \rho(i,j) \leq r \text{ or } j \in \C_i\}$, and $X_i^r = \E[Z_i \mid \F_i(r)]$. As a preliminary result, we bound the discrepancy between $Z_i$ and $X_i^s$.

  Let $t=\rho(\{i\},J(0,\C_i))$, the distance between $i$ and the nearest unit in the boundary of $\C_i$. By \autoref{aani}, for any $q>0$,
  \begin{multline*}
    \E[\abs{Z_i - X_i^s}^q \mid T_{1i}=1] = p_{1i}^{-q} \E[\abs{Y_i(\bm{D}) - \E[Y_i(\bm{D}) \mid \F_i(s)]}^q \mid T_{1i}=1] \\
    \leq p_{1i}^{-q} \psi(\max\{t,s\})^q.
  \end{multline*}

  \noindent The equality holds because $\F_i(s)$ conditions on $D_i$, and by \autoref{adesign}, $T_{1i}=1$ implies $D_i=1$, which implies $T_{0i}=0$. The inequality holds because $\F_i(s)$ fixes $\{D_j\colon j\in Q(i,t) \cup Q(i,s)\}$ at their realized values. Similarly, $\E[\abs{Z_i - X_i^s}^q \mid T_{0i}=1] \leq p_{0i}^{-q} \psi(\max\{t,s\})^q$, so by the law of total probability and \autoref{roverlap}, for some universal constant $c''>0$,
  \begin{equation}
    \E[\abs{Z_i - X_i^s}^q]^{1/q} \leq c'' \psi(\max\{t,s\}). \label{e1290urt2g0}
  \end{equation}

  Define $R_i = Z_i-X_i^{\kappa_n}$. Since $j\not\in\Lambda_i$, $\abs{\cov(X_i^{\kappa_n},X_j^{\kappa_n})} = 0$ by \autoref{adesign}. Applying the Cauchy-Schwarz and Jensen's inequalities and \eqref{e1290urt2g0} for $q=2$,
  \begin{align*}
    \abs{\cov(Z_i,Z_j)} &\leq \abs{\cov(X_i^{\kappa_n},X_j^{\kappa_n})} + \abs{\cov(X_i^{\kappa_n},R_j)} + \abs{\cov(R_i,X_j^{\kappa_n})} + \abs{\cov(R_i,R_j)} \\
			&\leq 2c'' (\norm{Z_i}_2 \psi(\max\{t,\kappa_n\}) + \norm{Z_j}_2 \psi(\max\{t,\kappa_n\}) + \psi(\max\{t,\kappa_n\})^2) \\
			&\leq c\,\psi(\max\{t,\kappa_n\})
  \end{align*} 

  \noindent for some universal $c>0$.

  Next consider the case $s > 4r_n$. Abbreviate $X_i = X_i^{s/2-r_n}$, and redefine $R_i = Z_i-X_i$. Note that $\rho(Q(i,s/2-r_n),Q(j,s/2-r_n)) > 2r_n$ and $2r_n$ is the diameter of a cluster, so $X_i \indep X_j$. Consequently, following the previous argument,
  \begin{multline}
    \abs{\cov(Z_i,Z_j)} \leq \abs{\cov(X_i,X_j)} + \abs{\cov(X_i,R_j)} + \abs{\cov(R_i,X_j)} + \abs{\cov(R_i,R_j)} \\
    \leq c\,\psi(\max\{t,s/2-r_n\}). \label{covzz0}
  \end{multline}

  {\em Step 2.} For any $c \in \R$, let $\lfloor c \rfloor$ denote $c$ rounded down to the nearest integer. Using the covariance bounds derived in step 1, 
  \begin{multline}
    \frac{1}{n^2} \sum_{i\in\N_n} \sum_{j \not\in \Lambda_i} \abs{\cov(Z_i,Z_j)} \leq \frac{c}{n^2} \sum_{k=1}^{m_n} \sum_{s=1}^{\lfloor 2R_n \rfloor} \sum_{t=0}^{\lfloor \min\{s,r_n-1\} \rfloor} \sum_{i\in J(t,C_k)} \sum_{j\not\in\Lambda_i} \bm{1}\{\rho(i,j)\in [s-1,s)\} \\
    \times \big( \psi(\max\{t,r_n/2\}) \bm{1}\{s \leq 4r_n\} + \psi(\max\{t,s/2-r_n\}) \bm{1}\{s > 4r_n\} \big) \\
    \equiv [B1] + [B2], \label{covzz}
  \end{multline} 

  \noindent where $[B1]$ takes the part involving $s \leq 4r_n$ and $[B2]$ the remainder. Note that $t$ can be at most $r_n-1$ because $J(r_n-1,C_k)$ is the 1-ball centered at the centroid of $C_k$, and it can be at most $s$ because $\rho(i,j)\in [s-1,s)$ and $j\not\in C_k$ since $j\not\in\Lambda_i$.

  As discussed at the start of this section, $\sum_{j\not\in\Lambda_i} \bm{1}\{\rho(i,j)\in [s-1,s)\} \leq \sum_{j\in\N_n} \bm{1}\{\rho(i,j)\in [s-1,s)\} \leq c'\,s$ and $\abs{J(t,C_k)} \leq c'(r_n-t)$ for some universal $c'>0$ by \autoref{aid}. Then by \autoref{aani},
  \begin{align}
    [B1] &\leq c\, \frac{m_n}{n^2} \sum_{s=1}^{\lfloor 4r_n \rfloor} c's \sum_{t=0}^{\lfloor \min\{s,r_n-1\} \rfloor} c'(r_n-t) \psi(t) \nonumber\\
	 &\lesssim \frac{m_n}{n^2} \sum_{s=1}^{\lfloor 4r_n \rfloor} s\left( r_n \sum_{t=0}^\infty \psi(t) + \sum_{t=0}^\infty t\,\psi(t) \right) 
    \lesssim \frac{m_n}{n^2} r_n^2 r_n \lesssim \frac{1}{\sqrt{nm_n}}. \label{B1}
  \end{align} 

  \noindent Finally,
  \begin{align}
    [B2] &\leq c\, \frac{m_n}{n^2} \sum_{s=\lfloor 4r_n \rfloor}^{\lfloor 2R_n \rfloor} c's \sum_{t=0}^{\lfloor \min\{s,r_n-1\} \rfloor} c'(r_n-t) \psi(s/2-r_n) \nonumber\\
	 &\lesssim \frac{m_n}{n^2} \sum_{t=0}^{\lfloor r_n-1 \rfloor} (r_n-t) \sum_{s=\lfloor 4r_n \rfloor}^{\lfloor 2R_n \rfloor} s\, \psi(s/2-r_n) 
    \lesssim \frac{m_n}{n^2} r_n^2 \sum_{s=\lfloor 4r_n \rfloor}^{\lfloor 2R_n \rfloor} s\,\psi(s/2-r_n) \nonumber\\
	 &\lesssim \frac{m_n}{n^2} r_n^2 r_n \lesssim \frac{1}{\sqrt{nm_n}}. \label{B2}
  \end{align}
\end{proof}

\begin{proof}[Proof of \autoref{tclt}]
  Recall that $\F_i = \{D_j\colon j \in \C_i \cup \N(i,\kappa_n)\}$, and define $R_i = Z_i - \E[Z_i \mid \F_i]$. 

  {\em Step 1.} We show that
  \begin{equation*}
    \E\left[ \left( \sqrt{m_n} \frac{1}{n} \sum_{i\in\N_n} R_i \right)^2 \right] = o(1).
  \end{equation*}

  \noindent Recalling \eqref{Lambdai}, the left-hand side equals
  \begin{equation*}
    \frac{m_n}{n^2} \sum_{i\in\N_n} \sum_{j\in\Lambda_i} \E[R_iR_j] + \frac{m_n}{n^2} \sum_{i\in\N_n} \sum_{j\not\in\Lambda_i} \E[R_iR_j] \equiv [C] + [D].
  \end{equation*}

  \noindent If $j\in\Lambda_i$, then $\rho(i,j) \leq 3r_n$. Using this fact and \eqref{e1290urt2g0} with $s=r_n/2$,
  \begin{align*}
    \abs{[C]} &\leq \frac{m_n}{n^2} \sum_{i\in\N_n} \E[R_i^2] + \frac{m_n}{n^2} \sum_{i\in\N_n} \sum_{s=1}^{\lfloor 3r_n \rfloor} \sum_{j\neq i} (c''\psi(r_n/2))^2 \bm{1}\{\rho(i,j)\in [s-1,s)\} \\ 
	      &\lesssim \frac{m_n}{n} + \frac{m_n}{n} \sum_{s=1}^{\lfloor 3r_n \rfloor} s\, \psi(r_n/2)^2 \lesssim \frac{m_n}{n} + \psi(r_n/2)^2 = o(1).
  \end{align*} 

  For $\abs{[D]}$, we first establish some covariance bounds. Fix $i,j$ such that $j\not\in\Lambda_i$, and let $s=\rho(i,j)$. Let $t=\rho(\{i\},J(0,\C_i))$, the distance between $i$ and the nearest unit in the boundary of $\C_i$. By \autoref{aani}, 
  \begin{equation*}
    \E[\abs{R_i}^2 \mid T_{1i}=1] = p_{1i}^{-2} \E\big[\abs{Y_i(\bm{D}) - \E[Y_i(\bm{D}) \mid \F_i]}^2 \mid T_{1i}=1\big] \leq p_{1i}^{-2} \psi(t)^2.
  \end{equation*} 

  \noindent Similarly, $\E[\abs{R_i}^2 \mid T_{0i}=1] \leq p_{0i}^{-2} \psi(t)^2$, so by the law of total probability and \autoref{roverlap}, for some universal constant $c>0$,
  \begin{equation*}
    \E[\abs{R_i}^2] \leq c\, \psi(t)^2 \quad\text{and}\quad \abs{\cov(R_i,R_j)} \leq c\,\psi(t). 
  \end{equation*}

  We derive an alternate bound for the case $s > 4r_n$. Let $\F_i(r) = \{D_j\colon \rho(i,j) \leq r \text{ or } j \in \C_i\}$ and $X_i = \E[Z_i \mid \F_i(s/2-r_n)]$. Then 
  \begin{equation*}
    S_i \equiv R_i - \E[R_i \mid \F_i(s/2-r_n)] = \tilde R_i \equiv Z_i - X_i
  \end{equation*}

  \noindent since $\N(i,s/2-r_n) \supseteq \N(i,\kappa_n)$. Moreover, $X_i \indep X_j$ since $\rho(Q(i,s/2-r_n),Q(j,s/2-r_n)) > 2r_n$ and $2r_n$ is the diameter of a cluster. Consequently, by \eqref{covzz0},
  \begin{align*}
    \abs{\cov(R_i,R_j)} &\leq \abs{\cov(X_i,X_j)} + \abs{\cov(X_i,S_j)} + \abs{\cov(S_i,X_j)} + \abs{\cov(S_i,S_j)} \\
			&= \abs{\cov(X_i,\tilde R_j)} + \abs{\cov(\tilde R_i,X_j)} + \abs{\cov(\tilde R_i,\tilde R_j)} \leq c\,\psi(s/2-r_n). 
  \end{align*}

  \noindent Applying the covariance bounds,
  \begin{multline*}
    \abs{[D]} \leq \frac{m_n}{n^2} \sum_{i\in\N_n} \sum_{j \not\in \Lambda_i} \abs{\E[R_iR_j]} \\ \leq c\frac{m_n}{n^2} \sum_{k=1}^{m_n} \sum_{s=1}^{\lfloor 2R_n \rfloor} \sum_{t=0}^{\lfloor \min\{s,r_n-1\} \rfloor} \sum_{i\in J(t,C_k)} \sum_{j\not\in\Lambda_i} \bm{1}\{\rho(i,j)\in [s-1,s)\} \\
    \times \big( \psi(t) \bm{1}\{s \leq 4r_n\} + \psi(s/2-r_n) \bm{1}\{s > 4r_n\} \big),
  \end{multline*}

  \noindent which is order $(m_n/n)^{1/2} = o(1)$ by an argument similar to \eqref{B1} and \eqref{B2}.

  {\em Step 2.} We show that
  \begin{equation*}
    \sigma_n^{-1} \sqrt{m_n} \frac{1}{n} \sum_{i\in\N_n} \big( \E[Z_i \mid \F_i] - \E[Z_i] \big) \dlimarrow \N(0,1).
  \end{equation*}

  \noindent First, let $\tilde\sigma_n^2 = \var(\sqrt{m_n}n^{-1} \sum_{i\in\N_n} \E[Z_i \mid \F_i])$. By Minkowski's inequality,
  \begin{equation}
    \abs{\sigma_n - \tilde\sigma_n} \leq \var\left( \sqrt{m_n} \frac{1}{n} \sum_{i\in\N_n} (Z_i - \E[Z_i \mid \F_i]) \right)^{1/2}, \label{mink}
  \end{equation}

  \noindent which is $o(1)$ by step 1. By \autoref{anondeg}, it then suffices to show
  \begin{equation}
    \tilde\sigma_n^{-1} \sqrt{m_n} \frac{1}{n} \sum_{i\in\N_n} \big( \E[Z_i \mid \F_i] - \E[Z_i] \big) \dlimarrow \N(0,1). \label{dpclt}
  \end{equation}

  We apply \autoref{ross} with $X_i = n^{-1}m_n^{1/2} (\E[Z_i \mid \F_i] - \E[Z_i])$ and dependency graph $\bm{A}$ defined after \eqref{decomp}. The maximum degree of $\bm{A}$ is at most $16\max_k \abs{C_k}$, 
  and $\max_k \abs{C_k} = O(n/m_n)$. Therefore, by Assumptions \ref{aboundY} and \ref{anondeg},
  \begin{equation*}
    \eqref{wass} \lesssim \left( \frac{n}{m_n} \right)^2 n \left( \frac{\sqrt{m_n}}{n} \right)^3 + \left( \frac{n}{m_n} \right)^{3/2} \sqrt{n\left( \frac{\sqrt{m_n}}{n} \right)^4} \lesssim m_n^{-1/2}.
  \end{equation*}

  \noindent Since $m_n\rightarrow\infty$, \eqref{dpclt} follows.
\end{proof}

\begin{seclemma}[\cite{ross2011fundamentals}, Theorem 3.6]\label{ross}
  Let $\{X_i\}_{i=1}^n$ be random variables with dependency graph $\bm{A}$ such that $\E[X_i^4]<\infty$ and $\E[X_i]=0$. Define $\sigma^2 \equiv \var(\sum_{i=1}^n X_i)$, $\mathcal{W} = \sum_{i=1}^n X_i/\sigma^2$, and $\Psi = \max_{i=1,\dots,n} \sum_{j=1}^n A_{ij}$. For $\mathcal{Z} \sim \mathcal{N}(0,1)$,
  \begin{equation}
    d(\mathcal{W},\mathcal{Z}) \leq \frac{\Psi^2}{\sigma^3} \sum_{i=1}^n \E[\abs{X_i}^3] + \frac{\sqrt{28} \Psi^{3/2}}{\sqrt{\pi} \sigma^2} \left( \sum_{i=1}^n \E[X_i^4] \right)^{1/2}, \label{wass}
  \end{equation}

  \noindent where $d(\cdot,\cdot)$ is the Wasserstein distance.
\end{seclemma}

\end{appendix}
%
%

\begin{acks}[Acknowledgments]
  The author thanks the referees and associate editor for helpful comments that improved the exposition of the paper.
\end{acks}
%

\begin{supplement}
\stitle{supplement.zip}
\sdescription{This zip file contains a PDF with proofs omitted in this text and code used to produce the simulation results in \autoref{smc}.}
\end{supplement}


\bibliographystyle{imsart-number} 
\bibliography{../scr}       



\ifarXiv
    

\section*{Supplementary material}

# SUPPLEMENT TO "RATE-OPTIMAL CLUSTER-RANDOMIZED DESIGNS FOR SPATIAL INTERFERENCE"

## BY MICHAEL P. LEUNG[1]

[1]*Department of Economics, UC Santa Cruz, leungm@ucsc.edu*

This supplement contains proofs omitted in [1].

PROOF OF THEOREM 2. Let $\mathcal{N}_n$ equal $Q(\mathbf{0}, 0.5\sqrt{n}) \cap \mathbb{Z}^2$, and arbitrarily assign half of the units at the coincidental boundaries of two clusters to each of the clusters ($|\mathcal{N}_n|$ may not exactly equal $n$, but we may add additional units just outside of $Q(\mathbf{0}, 0.5\sqrt{n})$ to make up the difference). Let $\psi'(\cdot)$ denote the derivative of $\psi(\cdot)$. We construct potential outcomes as follows:

$$Y_i(\boldsymbol{d}) = \left(\prod_{j \in Q(i, r_n/2)} d_j\right) \left(\sum_{s=1}^{2R_n} |\psi'(s)| \frac{\sum_{j \in \mathcal{N}_n} d_j \mathbf{1}\{\rho(i,j) = s\}}{\sum_{j \in \mathcal{N}_n} \mathbf{1}\{\rho(i,j) = s\}}\right)$$

for $R_n = 0.5\sqrt{n}$. Note that the sum in the denominator is at always least one.

Assumption 1 holds since

$$\sum_{s=1}^{\infty} |\psi'(s)| \leq |\psi'(1)| + \int_1^{\infty} |\psi'(s)| \, \mathrm{d}s = |\psi'(1)| + \psi(1) < \infty.$$

To verify Assumption 3, for any $s \geq 0$, fix $\boldsymbol{d}, \boldsymbol{d}' \in \{0,1\}^n$ such that $d_j = d_j'$ for all $j \in \mathcal{N}(i,s)$. If it is not the case that $d_j = 1$ for all $j \in Q(i, r_n/2)$, then $|Y_i(\boldsymbol{d}) - Y_i(\boldsymbol{d}')| = 0$, so assume it is the case. Then

$$|Y_i(\boldsymbol{d}) - Y_i(\boldsymbol{d}')| = \left|\sum_{t=\max\{s, r_n/2\}}^{2R_n} |\psi'(t)| \frac{\sum_{j \in \mathcal{N}_n}(d_j - d_j')\mathbf{1}\{\rho(i,j) = t\}}{\sum_{j \in \mathcal{N}_n} \mathbf{1}\{\rho(i,j) = t\}}\right|$$

$$\leq \sum_{t=s}^{2R_n} |\psi'(t)| \leq |\psi'(s)| + \int_s^{\infty} |\psi'(t)| \, \mathrm{d}t = |\psi'(s)| + \psi(s),$$

and $\sum_{s=1}^{\infty} s\left(|\psi'(s)| + \psi(s)\right) < \infty$, as desired.

We next derive the lower bound on the bias. For any $c \in \mathbb{R}$, let $\lceil c \rceil$ ($\lfloor c \rfloor$) denote $c$ rounded up (down) to the nearest integer. For $\Gamma_i = \{j \in \mathcal{N}_n : j \in C_k \text{ for some } k \text{ s.t. } \mathcal{N}(i, \kappa_n) \cap C_k \neq \varnothing\}$,

$$\mathbf{E}\left[Y_i\left(\frac{T_{1i}}{p_{1i}} - \frac{T_{0i}}{p_{0i}}\right)\right] - \left(Y_i(\mathbf{1}_n) - Y_i(\mathbf{0}_n)\right) = \mathbf{E}[Y_i(\boldsymbol{D}) \mid T_{1i} = 1] - Y_i(\mathbf{1}_n)$$

$$= \sum_{t=1}^{2R_n} |\psi'(t)| \frac{\sum_{j \notin \Gamma_i}(p-1)\mathbf{1}\{\rho(i,j) = t\}}{\sum_{j \in \mathcal{N}_n} \mathbf{1}\{\rho(i,j) = t\}} \equiv w_i.$$

The last line follows because, by Assumption 4, $T_{1i} = 1$ implies $D_j = 1$ for all $j \in \Gamma_i$, and for all $j \notin \Gamma_i$, $\mathbf{E}[D_j \mid T_{1i} = 1] = p$. By construction of $\mathcal{N}_n$, for all $t \in [\lceil 1.5r_n \rceil + 1, \lfloor R_n \rfloor]$,

(S.1) $$\liminf_{n \to \infty} \frac{1}{n} \sum_{i \in \mathcal{N}_n} \frac{\sum_{j \notin \Gamma_i} \mathbf{1}\{\rho(i,j) = t\}}{\sum_{j \in \mathcal{N}_n} \mathbf{1}\{\rho(i,j) = t\}} > 0.$$

To see this, consider any $i$ for which $i \in Q(\boldsymbol{c}(C_k), r_n/2)$ for some cluster $k$. That is, $i$ lies in the "interior" of some cluster. For any such $i$, a nontrivial share of units $\ell$ such that $\rho(i, \ell) = t$

1

for $t \geq \lceil 1.5r_n \rceil + 1$ satisfy $j \notin \Gamma_i$ in $\mathbb{Z}^2$. For example, this share is unity if $i = \boldsymbol{c}(C_k)$. Since a nontrivial share of units $i$ lie in the interior of some cluster, (S.1) follows. We consider $t \geq \lceil 1.5r_n \rceil + 1$ because for some clusters $k$ and $i \in Q(\boldsymbol{c}(C_k), r_n/2)$, the closest $j \notin \Gamma_i$ satisfies $\rho(i,j) \geq \lceil 1.5r_n \rceil + 1$.

We therefore have

$$|\mathbf{E}[\hat{\theta}] - \theta_n| = \left| \frac{1}{n} \sum_{i \in \mathcal{N}_n} w_i \right| \gtrsim \sum_{t=\lceil 1.5r_n \rceil + 1}^{\lfloor R_n \rfloor} |\psi'(t)|$$

$$\geq \int_{\lceil 1.5r_n \rceil + 1}^{\lfloor R_n \rfloor} |\psi'(t)| \, dt = \psi(\lceil 1.5r_n \rceil + 1) - \psi(\lfloor R_n \rfloor).$$

Since $r_n = R_n/\sqrt{m_n}$, the right-hand side is of order $\psi(1.5 r_n)$.

Lastly, we consider the variance: $\mathrm{Cov}(Z_i, Z_j)$ equals

$$\sum_{k \in \mathcal{N}_n} \sum_{\ell \in \mathcal{N}_n} \sum_{s=1}^{2R_n} \sum_{t=1}^{2R_n} \frac{|\psi'(s)| \mathbf{1}\{\rho(i,k)=s\}}{p_{1i} \sum_{k' \in \mathcal{N}_n} \mathbf{1}\{\rho(i,k')=s\}} \frac{|\psi'(t)| \mathbf{1}\{\rho(j,\ell)=t\}}{p_{1j} \sum_{\ell' \in \mathcal{N}_n} \mathbf{1}\{\rho(j,\ell')=t\}} \mathrm{Cov}(T_{1i} D_k, T_{1j} D_\ell).$$

Note that $\mathrm{Cov}(T_{1i} D_k, T_{1j} D_\ell)$ is uniformly bounded away from zero if $j \in \Lambda_i$ and is always non-negative by the design. Also, $p_{1i}$ and $p_{0i}$ are uniformly bounded away from zero by Remark 3, so $\mathrm{Var}(n^{-1} \sum_{i \in \mathcal{N}_n} Z_i)$ is bounded below by

$$\frac{c}{n^2} \sum_{i \in \mathcal{N}_n} \sum_{j \in \Lambda_i} \sum_{s=1}^{2R_n} \sum_{t=1}^{2R_n} \sum_{k \in \mathcal{N}_n} \sum_{\ell \in \mathcal{N}_n} \frac{|\psi'(s)| \mathbf{1}\{\rho(i,k)=s\}}{\sum_{k' \in \mathcal{N}_n} \mathbf{1}\{\rho(i,k')=s\}} \frac{|\psi'(t)| \mathbf{1}\{\rho(j,\ell)=t\}}{\sum_{\ell' \in \mathcal{N}_n} \mathbf{1}\{\rho(j,\ell')=t\}}$$

$$= \frac{c}{n^2} \sum_{i \in \mathcal{N}_n} |\Lambda_i| \sum_{s=1}^{2R_n} \sum_{t=1}^{2R_n} |\psi'(s)| |\psi'(t)|$$

for some universal constant $c > 0$. Since $\min_i |\Lambda_i| \gtrsim n/m_n$, $\mathrm{Var}(\hat{\theta}) \gtrsim m_n^{-1}$. $\square$

PROOF OF THEOREM 4. As shown in (B.7), $|\sigma_n - \tilde{\sigma}_n| = o(1)$, and by Assumption 4,

$$\tilde{\sigma}_n^2 = \frac{m_n}{n^2} \sum_{i \in \mathcal{N}_n} \sum_{j \in \mathcal{N}_n} \mathrm{Cov}(\mathbf{E}[Z_i \mid \mathcal{F}_i], \mathbf{E}[Z_j \mid \mathcal{F}_j]) A_{ij}.$$

Then by Assumption 5, it is enough to show

(S.2) $\quad \hat{\sigma}^2 - \mathcal{R}_n = \dfrac{m_n}{n^2} \displaystyle\sum_{i \in \mathcal{N}_n} \sum_{j \in \mathcal{N}_n} \mathrm{Cov}(\mathbf{E}[Z_i \mid \mathcal{F}_i], \mathbf{E}[Z_j \mid \mathcal{F}_j]) A_{ij} + o_p(1).$

The first two steps below will establish that

(S.3) $\quad \hat{\sigma}^2 - \mathcal{R}_n = \dfrac{m_n}{n^2} \displaystyle\sum_{i \in \mathcal{N}_n} \sum_{j \in \mathcal{N}_n} (Z_i - \mathbf{E}[Z_i])(Z_j - \mathbf{E}[Z_j]) A_{ij} + o_p(1).$

*Step 1.* In the formula of $\hat{\sigma}^2$, replace $\bar{Z}$ with $\bar{Z} - \mathbf{E}[\bar{Z}] + \mathbf{E}[\bar{Z}]$ to obtain

$$\hat{\sigma}^2 = \underbrace{\frac{m_n}{n^2} \sum_{i \in \mathcal{N}_n} \sum_{j \in \mathcal{N}_n} (Z_i - \mathbf{E}[\bar{Z}])(Z_j - \mathbf{E}[\bar{Z}]) A_{ij}}_{\check{\sigma}^2}$$

$$+ (\mathbf{E}[\bar{Z}] - \bar{Z}) \frac{2 m_n}{n^2} \sum_{i \in \mathcal{N}_n} (Z_i - \mathbf{E}[\bar{Z}]) \sum_{j \in \mathcal{N}_n} A_{ij} + (\mathbf{E}[\bar{Z}] - \bar{Z})^2 \frac{m_n}{n^2} \sum_{i \in \mathcal{N}_n} \sum_{j \in \mathcal{N}_n} A_{ij}.$$

The last two terms on the right-hand side are $o_p(1)$ since $\sum_{j\in\mathcal{N}_n} A_{ij} \leqslant 16\max_k |C_k| \lesssim n/m_n$, and $|\bar{Z} - \mathbf{E}[\bar{Z}]| \xrightarrow{p} 0$ by Theorem 1.

*Step 2.* In the formula of $\check{\sigma}^2$, replace $Z_i$ with $Z_i - \mathbf{E}[Z_i] + \mathbf{E}[Z_i]$ to obtain

$$\check{\sigma}^2 = \frac{m_n}{n^2} \sum_{i\in\mathcal{N}_n} \sum_{j\in\mathcal{N}_n} (Z_i - \mathbf{E}[Z_i])(Z_j - \mathbf{E}[Z_j])A_{ij}$$

$$+ \frac{2m_n}{n^2} \sum_{i\in\mathcal{N}_n} \sum_{j\in\mathcal{N}_n} (Z_i - \mathbf{E}[Z_i])(\mathbf{E}[Z_j] - \mathbf{E}[\bar{Z}])A_{ij} + \mathcal{R}_n.$$

We show that the second term on the right-hand side is $o_p(1)$. For $W_i \equiv \sum_{j\in\mathcal{N}_n}(\mathbf{E}[Z_j] - \mathbf{E}[\bar{Z}])A_{ij}$,

$$\mathbf{E}\left[\left|\frac{m_n}{n^2} \sum_{i\in\mathcal{N}_n} \sum_{j\in\mathcal{N}_n} (Z_i - \mathbf{E}[Z_i])(\mathbf{E}[Z_j] - \mathbf{E}[\bar{Z}])A_{ij}\right|\right]$$

$$\leqslant \mathbf{E}\left[\left(\frac{m_n}{n^2} \sum_{i\in\mathcal{N}_n} (Z_i - \mathbf{E}[Z_i])W_i\right)^2\right]^{1/2}$$

$$= \left(\frac{m_n^2}{n^4} \sum_{i\in\mathcal{N}_n} \sum_{j\in\Lambda_i} \mathrm{Cov}(Z_i, Z_j)W_i W_j + \frac{m_n^2}{n^4} \sum_{i\in\mathcal{N}_n} \sum_{j\notin\Lambda_i} \mathrm{Cov}(Z_i, Z_j)W_i W_j\right)^{1/2}$$

By Assumption 1 and Remark 3, $Z_i$ is uniformly bounded. Since $\max_i \sum_{j\in\mathcal{N}_n} A_{ij} \lesssim n/m_n$, $W_i$ is uniformly bounded by a constant times $n/m_n$. It follows from the variance bound for (5) that the last line is $o(1)$.

*Step 3.* Having proven (S.3) and the sufficiency of establishing (S.2), it remains to show that the following is $o_p(1)$:

$$\frac{m_n}{n^2} \sum_{i\in\mathcal{N}_n} \sum_{j\in\mathcal{N}_n} (Z_i - \mathbf{E}[Z_i])(Z_j - \mathbf{E}[Z_j])A_{ij}$$

$$- \frac{m_n}{n^2} \sum_{i\in\mathcal{N}_n} \sum_{j\in\mathcal{N}_n} \mathrm{Cov}(\mathbf{E}[Z_i \mid \mathcal{F}_i], \mathbf{E}[Z_j \mid \mathcal{F}_j])A_{ij}$$

$$= \frac{m_n}{n^2} \sum_{i\in\mathcal{N}_n} \sum_{j\in\mathcal{N}_n} (Z_i Z_j - \mathbf{E}[\mathbf{E}[Z_i \mid \mathcal{F}_i]\mathbf{E}[Z_j \mid \mathcal{F}_j]])A_{ij}$$

$$- 2\frac{m_n}{n^2} \sum_{i\in\mathcal{N}_n} (Z_i - \mathbf{E}[Z_i]) \underbrace{\sum_{j\in\mathcal{N}_n} \mathbf{E}[Z_j]A_{ij}}_{\leqslant cn/m_n}.$$

The last term on the right-hand side is $o_p(1)$ by Theorem 1. The first term equals

$$\frac{m_n}{n^2} \sum_{i\in\mathcal{N}_n} \sum_{j\in\mathcal{N}_n} (Z_i - \mathbf{E}[Z_i \mid \mathcal{F}_i])(Z_j - \mathbf{E}[Z_j \mid \mathcal{F}_j])A_{ij}$$

$$+ 2\frac{m_n}{n^2} \sum_{i\in\mathcal{N}_n} \sum_{j\in\mathcal{N}_n} \mathbf{E}[Z_i \mid \mathcal{F}_i](Z_j - \mathbf{E}[Z_j \mid \mathcal{F}_j])A_{ij}$$

$$+ \frac{m_n}{n^2} \sum_{i\in\mathcal{N}_n} \sum_{j\in\mathcal{N}_n} (\mathbf{E}[Z_i \mid \mathcal{F}_i]\mathbf{E}[Z_j \mid \mathcal{F}_j] - \mathbf{E}[\mathbf{E}[Z_i \mid \mathcal{F}_i]\mathbf{E}[Z_j \mid \mathcal{F}_j]])A_{ij}$$

$$\equiv [E] + [F] + [G].$$

We claim that $[E]$ and $[F]$ are $o_p(1)$. Note that $Z_i - \mathbf{E}[Z_i \mid \mathcal{F}_i]$ equals

$$\frac{T_{1i}}{p_{1i}}(Y_i(\boldsymbol{D}) - \mathbf{E}[Y_i(\boldsymbol{D}) \mid T_{1i} = 1, \mathcal{F}_i]) - \frac{T_{0i}}{p_{0i}}(Y_i(\boldsymbol{D}) - \mathbf{E}[Y_i(\boldsymbol{D}) \mid T_{0i} = 1, \mathcal{F}_i]).$$

Furthermore, for any $t \in \{0, 1\}$,

$$T_{ti}|Y_i(\boldsymbol{D}) - \mathbf{E}[Y_i(\boldsymbol{D}) \mid T_{ti} = 1, \mathcal{F}_i]| \leq \psi(r_n/2)$$

by Assumption 3 since $T_{ti} = 1$ implies $D_j = t$ for all $j \in Q(i, r_n/2)$. Then because, as previously discussed, $\max_i \sum_{j \in \mathcal{N}_n} A_{ij} \lesssim n/m_n$,

$$[E] \lesssim \frac{m_n}{n^2} \cdot n \cdot \frac{n}{m_n} \cdot \psi(r_n/2)^2 = o(1) \quad \text{and} \quad [F] \lesssim \frac{m_n}{n^2} \cdot n \cdot \frac{n}{m_n} \cdot \psi(r_n/2) = o(1).$$

Since the mean of $[G]$ is zero, it remains to bound its variance, which equals

$$(S.4) \quad \frac{m_n^2}{n^4} \sum_{i \in \mathcal{N}_n} \sum_{j \in \mathcal{N}_n} \sum_{k \in \mathcal{N}_n} \sum_{\ell \in \mathcal{N}_n} \text{Cov}(\mathbf{E}[Z_i \mid \mathcal{F}_i]\mathbf{E}[Z_j \mid \mathcal{F}_j], \mathbf{E}[Z_k \mid \mathcal{F}_k]\mathbf{E}[Z_\ell \mid \mathcal{F}_\ell])A_{ij}A_{k\ell}.$$

The covariance term is zero if $k, \ell \notin \Lambda_i \cup \Lambda_j$, so by Assumption 1 and Remark 3, there exists a universal constant $c > 0$ such that

$$(S.5) \quad (S.4) \leq c \frac{m_n^2}{n^4} \sum_{i \in \mathcal{N}_n} \sum_{j \in \Lambda_i} |\Xi_{ij}|$$

for $\Xi_{ij} = \{(k, \ell) \colon \ell \in \Lambda_k \text{ and either } k \in \Lambda_i \cup \Lambda_j, \ell \in \Lambda_i \cup \Lambda_j, \text{or both}\}$. Since $\max_i \Lambda_i \lesssim n/m_n$ and $\ell \in \Lambda_k$ implies $\rho(k, \ell) \leq 3r_n$,

$$\sup_{n \in \mathbb{N}} \max_{i,j \in \mathcal{N}_n} |\Xi_{ij}| \lesssim \underbrace{\frac{n}{m_n}r_n}_{\substack{\text{only } k \in \Lambda_i \cup \Lambda_j, \\ \ell \in \Lambda_k}} + \underbrace{r_n \frac{n}{m_n}}_{\substack{\text{only } \ell \in \Lambda_i \cup \Lambda_j, \\ \ell \in \Lambda_k}} + \underbrace{\frac{n^2}{m_n^2}}_{k,\ell \in \Lambda_i \cup \Lambda_j} \lesssim \frac{n^2}{m_n^2}.$$

Therefore,

$$(S.5) \lesssim \frac{m_n^2}{n^4} \cdot n \cdot \frac{n}{m_n} \cdot \frac{n^2}{m_n^2} \lesssim m_n^{-1}.$$

$\square$


\fi


\end{document}

\end{document}